\documentclass[runningheads,a4paper]{llncs}

\usepackage{amssymb}
\usepackage{amsmath}
\usepackage{bbm}
\usepackage{latexsym}
\usepackage{xspace}
\usepackage{smallsec}
\usepackage{times}
\usepackage{xfrac}
\usepackage{gastex}

\title{Minimizing Expected Cost Under Hard Boolean Constraints, with Applications to Quantitative Synthesis} 

\titlerunning{Minimizing Expected Cost Under Hard Boolean Constraints} %optional, in case that the title is too long; the running title should fit into the top page column

\authorrunning{S. Almagor, O. Kupferman, and Y. Velner} 

\author{Shaull Almagor \and Orna Kupferman \and Yaron Velner}
\institute{School of Computer Science and Engineering, The Hebrew University, Israel.}

\newtheorem{xmpl}[theorem]{Example}
\renewenvironment{example}{\begin{xmpl}\rm}{\end{xmpl}}

\def\squarebox#1{\hbox to #1{\hfill\vbox to #1{\vfill}}}
\newcommand{\short}[1]{}
\newcommand{\set}[1]{{\{#1\}}}
\newcommand{\Nat}{\mathbb{N}}

\newcommand{\zug}[1]{\langle #1  \rangle}

\newcommand{\LTL}{{\ensuremath{\rm LTL}}\xspace}

\newcommand{\Next}{\mathsf{X}}
\newcommand{\Ev}{\mathsf{F}}
\newcommand{\Alw}{\mathsf{G}}

\newcommand{\T}{{\mathcal T}}

\newcommand{\DPW}{\mbox{\rm DPW}\xspace}

\newcommand{\UPW}{\mbox{\rm UPW}\xspace}

\newcommand{\stam}[1]{}

\newcommand{\C}{{\cal C}}
\newcommand{\A}{{\cal A}}
\newcommand{\D}{{\cal D}}
\newcommand{\G}{{\cal G}}
\newcommand{\M}{{\cal M}}

\newcommand{\U}{{\cal U}}

\renewcommand{\P}{{\mathtt P}}
\newcommand{\MDP}{{\mathtt{MDP}}}

\renewcommand{\phi}{\varphi}

\newcommand{\maxs}[1]{\max\set{#1}}

\newcommand{\Inf}{\mbox{inf}}

\newcommand{\sen}{{\it sensed}}
\newcommand{\scost}{{\it scost}}
\newcommand{\stcost}{{\it scost}}
\newcommand{\cost}{{\it cost}}

\newcommand{\tIN}{{2^I}}
\newcommand{\tOUT}{{2^O}}
\newcommand{\tINs}{{{(2^I)}^*}}

\newcommand{\tINo}{{{(2^I)}^\omega}}
\newcommand{\tOUTo}{{{(2^O)}^\omega}}

\newcommand{\Act}{A}
\newcommand{\MDPProb}{{\rm P}}
\newcommand{\MDPcost}{{\it cost}}

\newcommand{\start}{{\textsc{start}}}

\newcommand{\buchi}{B\"uchi\xspace}

\newcommand{\costs}{{\it cost}_{\rm sure}}
\newcommand{\costsf}{{\it cost}_{\rm sure,<\infty}}

\newcommand{\real}{{\tt real}}

\newcommand{\gec}{{GEC}\xspace}
\newcommand{\sgec}{{SGEC}\xspace}
\newcommand{\gecs}{{GECs}\xspace}
\newcommand{\sgecs}{{SGECs}\xspace}
\newcommand{\attr}{{\rm Attr}}
\newcommand{\env}{{\rm env}}
\newcommand{\sys}{{\rm sys}}

\newcommand{\me}{{\substack{\rm max\\ \rm even}}}

\begin{document}
\maketitle
\begin{abstract}
In Boolean synthesis, we are given an \LTL specification, and the goal is to construct a transducer that realizes it against an adversarial environment. 
%In quantitative synthesis, the specification is multi-valued, mapping each computation to its cost or utility. Then, in the stochastic setting, the behavior of the environment is probabilistic, and the goal is to construct a transducer that minimizes the expected cost (or maximizes the expected utility).
%This reduces quantitative synthesis to a solution of mean-payoff MDPs.
Often, a specification contains both Boolean requirements that should be satisfied against an adversarial environment, and multi-valued components that refer to the quality of the satisfaction and whose expected cost we would like to minimize with respect to a probabilistic environment. 

In this work we study, for the first time, mean-payoff games in which the system aims at minimizing the expected cost against a probabilistic environment, while surely satisfying an $\omega$-regular condition against an adversarial environment.
We consider the case the $\omega$-regular condition is given as a parity objective or by an \LTL formula.
We show that in general, optimal strategies need not exist, and moreover, the limit value cannot be approximated by finite-memory strategies. 
We thus focus on computing the limit-value, and give tight complexity bounds for synthesizing $\epsilon$-optimal strategies for both finite-memory and infinite-memory strategies.

We show that our game naturally arises in various contexts of synthesis with Boolean and multi-valued objectives. Beyond direct applications, in synthesis with costs and rewards to certain behaviors, it allows us to compute the minimal sensing cost of $\omega$-regular specifications -- a measure of quality in which we look for a transducer that minimizes the expected number of signals that are read from the input.
\end{abstract}

\section{Introduction}
\label{sec:intro}
{\em Synthesis\/} is the automated construction of a system from its specification: given a linear temporal logic (LTL)  formula $\psi$ over sets $I$ and $O$ of input and output signals, we synthesize a 
%finite-state 
system that {\em realizes\/} $\psi$ \cite{Chu63,PR89a}. At each moment in time, the system reads a truth assignment, generated by the environment, to the signals in $I$, and it generates a truth assignment to the signals in $O$. Thus, with every sequence of inputs, the system associates a sequence of outputs. 
The system realizes $\psi$ if all the computations that are generated by the interaction satisfy $\psi$. 
%Synthesis has attracted a lot of research and interest \cite{Kup12}.

One weakness of automated synthesis in practice is that it pays no attention to the quality of the synthesized system. Indeed, the classical setting is Boolean: a computation satisfies a specification or does not satisfy it. Accordingly, while the synthesized system is correct, there is no guarantee about its quality. This is a crucial drawback, as designers would be willing to give-up manual design only if automated-synthesis algorithms return systems of comparable quality. 
In recent years, researchers have considered extensions of the classical Boolean setting to a quantitative one, which takes quality into account. Quality measures can refer to the system itself, examining parameters like its size or its consumption of memory, sensors, voltage, bandwidth, etc., or refer to the way the system satisfies the specification. In the latter, we allow the designer to specify the quality of a behavior using quantitative specification formalisms~\cite{ABK16,BMM14,AFHMS05}. 
For example, rather than the Boolean specification requiring all requests to be followed by a grant, a quantitative specification formalism would give a different satisfaction value to a computation in which requests are responded immediately and one in which requests are responded after long delays.\footnote{Note that the polarity of some quality measures is negative, as we want to minimize size, consumption, costs, etc., whereas the polarity of other measures is positive,  as we want to maximize performance and satisfaction value. For simplicity, we assume that all measures are associated with costs, which we want to minimize.}

Solving the synthesis problem in the Boolean setting amounts to solving a two-player zero-sum game between the system and the environment. The goal of the system is to satisfy the (Boolean) specification, and the environment is adversarial. Then, a winning strategy for the system corresponds to a transducer that realizes the specification. 
In the quantitative setting, the goal of the system is no longer Boolean, as every play is assigned a cost by the specification. In the classical quantitative approach, we measure the satisfaction value in the worst-case semantics. Thus, the value of a strategy for the system is the maximal cost of a play induced by this strategy, and the goal of the system is to minimize this value. Recently, there is a growing interest also in the expected cost of a play, under a probabilistic environment. The motivation behind this approach is that the quality of satisfaction is a ``soft constraint'', and should not be measured in a worst-case semantics. Then, the game above is replaced by a mean-payoff Markov Decision Process (MDP): a game in which each state has a cost, inducing also costs to infinite plays (essentially, the cost of an infinite play is the limit of the average cost of prefixes of increased lengths). The goal is to find a strategy that minimizes the expected cost \cite{CKK15,CR15}.

While quantitative satisfaction refines the Boolean one, often a specification contains both Boolean conditions that should be satisfied against all environments, and multi-valued components that refer to the quality of the satisfaction and whose expectation we would like to minimize with respect to a probabilistic environment. Accordingly, 
researchers have suggested the {\em beyond worst-case} approach, where a specification has both hard and soft constraints, and the goal is to realize the hard constraints, while maximizing the expected satisfaction value of the soft constraints. In Section~\ref{related} below, we describe this approach and related work in detail.

In this work, we consider, for the first time, mean-payoff MDPs equipped with a parity winning condition (parity-MDPs, for short). The goal is to find a strategy that surely wins the parity game (that is, against an adversarial environment), while minimizing the expected cost of a play against a probabilistic environment. 
While the starting point in earlier related work is the MDP itself, possibly augmented by different objectives, 
our starting point depends on the application, and we view the construction of the MDP as an integral part of our contribution. We focus on two applications: synthesis with penalties to undesired scenarios and synthesis with minimal sensing.

Let us describe the two applications. We start with penalties to scenarios.
Consider an $\LTL$ specification $\psi$ over $I$ and $O$. Activating an output signal may have a cost; for example, when the activation involves a use of a resource. Taking these costs into account, the input to the synthesis problem includes, in addition to $\psi$, a cost function $\gamma$ assigning cost to some assignments to output signals. The cost of a computation is then the mean cost of assignments in it. While the specification $\psi$ is a hard constraint, as we only allow correct computations, minimizing the expected cost of computations with respect to $\gamma$ is a soft constraint. 
Assignments correspond to scenarios of length one. More elaborated cost functions refer to on-going regular scenarios. 
Power consumption, for example, is an important consideration in modern chip design, from portable servers
to large server farms. As the chips become more complex, the cost of powering a server
farm can easily outweigh the cost of the servers themselves, thus design teams go to great
lengths in order to reduce power consumption in their designs.
%Existing power saving techniques can be divided into electrical, such as using more efficient transistors, and logical, which attempt to introduce power-saving changes into designs without changing their logic. 
%One approach to power-saving refers to activation of signals, discussed above. As it turns out, however, 
%logical power-saving techniques mainly attempt to reduce the number of changes in the values of signals, the main source of power consumption in chips. 
The most widely researched logical power saving techniques are \emph{clock gating}, in which a clock is prevented from making a ``tick'' if it is redundant (c.f., \cite{ARY09}), and \emph{power gating}, in which whole sections of the chip are powered off when not needed and then powered on again \cite{KFAGS07,ENY09}. The goal of these techniques is to reduce power consumption and the number of changes in the values of signals, the main source of power consumption in chips. 
The input to the problem of synthesis with penalties to scenarios includes, in addition to $\psi$, a set of deterministic automata on finite words, each describing a undesired scenario and its cost. For example, it is easy to specify the scenario of ``value flip" with a two-state deterministic automaton.
We show how the setting can be easily translated into solving our parity-MDPs, thus generating systems that realize $\psi$ with minimal expected cost. 
%Moreover, the complexity of the solution is 2EXPTIME, matching the complexity of Boolean synthesis.

Our primary application considers activation of sensors. The quality measure of sensing was introduced in ~\cite{AKK14,AKK15}, as a measure for the detail with which a random input word needs to be read in order to realize the specification. In the context of synthesis, our goal is to construct a transducer that realizes the specification and minimizes the expected average number of sensors (of input signals) that are used along the interaction.
Thus, the hard constraint in the LTL specification, and the soft one is the expected number of active sensors. 
Giving up sensing has a flavor of synthesis with incomplete information \cite{KV00a}: the transducer has to realize the specification no matter what the incomplete information is. Thus, as opposed to the examples above, the modeling of cost involves a careful construction of the MDP to be analyzed, and also involves an exponential blow-up, which we show to be unavoidable.  In \cite{AKK15}, the problem was solved for safety specifications. Our solution to the parity-MDP problem enables a solution for full LTL. We also study the complexity of the problem when the input is an LTL formula, rather than a deterministic automaton. 
%, and show
%%shaull4
%that we can avoid one exponent by combining the translation of $\LTL$ to automata with our construction for the sensing cost. Thus, computing the sensing cost of an \LTL specification can be done in doubly-exponential time, which matches the complexity of Boolean synthesis.

Back to parity-MDPs, we show that in general, optimal strategies need not exist. That is, there are parity-MDPs in which an infinite-state strategy can get arbitrarily close to some limit optimal value, but cannot attain it. Moreover, the limit value cannot be approximated by finite-memory strategies.  Accordingly, our solution to parity-MDPs suggests two algorithms. The first, described in Section~\ref{sec:infinite memory}, finds the limit value of all possible strategies, which corresponds to infinite-state transducers. The second, described in Section~\ref{sec:finite memory}, computes the limit value over all finite-memory strategies. 
The complexity of both algorithms is NP$\cap$coNP. Moreover, they are computable in polynomial time when an oracle to a two-player parity game is given.
Hence, our complexity upper bounds match the trivial lower bounds that arise from the fact that every solution to a parity-MDP is also a solution to a parity game.
For our applications, we show that the complexity of the synthesis problem for LTL specifications stays doubly-exponential, as in the Boolean setting, even when we minimize penalties to undesired scenarios or minimize  sensing.

\vspace*{-10pt}
\subsection{Related Work}
\label{related}
%shaull1
The combination of worst-case synthesis with expected-cost synthesis, dubbed {\em beyond worst-case synthesis\/}, was studied in~\cite{BFRR14a,CR15} for models that are closely related to ours. 
%Our work follows their line of research. 
In~\cite{BFRR14a} the authors study mean-payoff MDPs, where both the hard constraints and the soft constraints are quantitative. Thus, a system needs to ensure a strict upper bound on the mean-payoff cost, while minimizing the expected cost. 
In~\cite{CR15}, multidimensional mean-payoff MDPs are considered. Thus, the MDP is equipped with several mean-payoff costs, and the goal is to find a system that ensures the mean-payoff in some of the mean-payoffs is below an upper bound, while minimizing the expected mean-payoffs (or rather, approximating their Pareto-curve).
 
In comparison, our work is the first to consider a hard Boolean constraint (namely the parity condition). This poses both a conceptual and a technical difference. Conceptually, when quantitative synthesis is taken as a refinement of Boolean synthesis, it is typically meant as a ranking of different systems that satisfy a Boolean specification. Thus, it makes sense for the hard constraint to be Boolean as well. 
Technically, combining Boolean and quantitative constraints gives rise to some subtleties that do not exist in the pure-quantitative setting. Specifically, when both the hard and the soft constraints are quantitative, a strategy can intuitively ``alternate'' between satisfying them. Thus, if while trying to meet the soft constraint the hard constraint is violated, we can switch to a worst-case strategy until the hard constraint is satisfied, and go back to trying to minimize the soft constraint. This alternation can be done infinitely often. In the Boolean setting, however, this alternation can violate the Boolean constraint. We note that 
%shaull7
unlike classical parity games, where the parity winning condition can be translated to a richer mean-payoff objective,
the parity winning condition in our parity-MDPs does not admit a similar translation.
%cannot be easily translated to a richer mean-payoff objective, yak yak

Other works on MDPs and mean-payoff objectives tackle different aspects of quantitative analysis. In~\cite{Put14}, a solution to the expected mean-payoff value over MDPs is presented. In~\cite{CD11} and~\cite{CLGO14}, the authors study a combination of mean-payoff and parity objectives over MDPs and over stochastic two-player games. There, the goal of the system is to ensure with probability 1 that the parity condition holds and that the mean-payoff is below a threshold. This differs from our work in that the parity condition is not a hard constraint, as it is met only almost-surely, and in that the 
%shaull9
%mean-payoff might not be 
expected mean-payoff is not guaranteed to be
minimized. As detailed in the paper, these differences make the technical challenges very different.

Due to lack of space, some proofs appear in the appendix.
\section{Parity-MDPs}
\label{sec:prelim}

%\subsubsection*{Parity Markov Decision Processes} 
\stam{
A Markov chain $\M=\zug{S,P}$ consists of a finite state space $S$ and a 
  stochastic transition matrix $P:S \times S \rightarrow [0,1]$. That is, for all $s \in S$, we have $\sum_{s' \in S} P(s,s') =1$. Given an initial state $s_0$, consider the vector $v^0$ in which $v^0(s_0)=1$ and $v^0(s)=0$ for every $s\neq s_0$. The {\em limiting distribution} of $\M$ is $\lim_{n\to \infty}\frac{1}{n}\sum_{m=0}^n v^0P^m$. The limiting distribution satisfies $\pi P=\pi$, and can be computed in polynomial time~\cite{GL97}.
}

A {\em parity Markov decision process\/} (Parity-MDP, for short) combines a parity game with a mean-payoff MDP. The game is played between Player~1, who models a 
%shaull6
%deterministic
system, and Player~2, who models the environment. The environment is adversarial with respect to the parity winning condition and is stochastic with respect to the mean-payoff objective. Formally, a parity-MDP is a tuple
$\M=\zug{S_1, S_2,s_0,\Act_1,\Act_2,\delta_1, \delta_2,\MDPProb,\MDPcost,\alpha}$, with the following components. The sets $S_1$ and $S_2$ are finite set of states, for Players 1 and 2, respectively. Let $S=S_1 \cup S_2$. Then, $s_0\in S$ is an initial state, and $\Act_1$ and $\Act_2$ are sets of actions for the players. Not all actions are available in all states: for every state $s\in S_i$, for $i\in \set{1,2}$, we use $\Act_i(s)$ to denote the finite set of actions available to Player $i$ in the state $s$. 
For $i\in \set{1,2}$, the transition function $\delta_i:S_i\times \Act_i\nrightarrow S$ is such that $\delta_i(s,a)$ is defined iff $a\in \Act_i(s)$. Let $\delta=\delta_1 \cup \delta_2$. 
Note that $\delta_2$ gets an action of Player~2 as a parameter. We distinguish between two approaches to the way the action is chosen. In the adversarial approach, it is Player~2 who chooses the action. In the stochastic approach, the choice depends on the (partial) function $\MDPProb:S_2\times \Act_2\nrightarrow [0,1]$, where for every state $s\in S_2$ and $a \in A_2$, we have that $ \MDPProb(s,a)>0$ only if $a \in A_2(s)$. Also, $\sum_{a\in \Act_2(s)} \MDPProb(s,a)=1$. Finally, $\MDPcost:S\to \Nat$ is a cost function, and  $\alpha:S\to \set{0,...,d}$, for some $d \in \Nat$, is a parity winning condition.

The parity-MDP  $\M$ induces a {\em parity game\/} $\M^{\P}=\zug{S_1, S_2,s_0,\Act_1,\Act_2,\delta_1, \delta_2,\alpha}$, obtained by omitting $\MDPProb$ and $\MDPcost$. In this game, we follow the adversarial approach to the environment. Thus, both players choose their actions. Formally, 
a {\em strategy} for Player $i$  in $\M$, for $i \in \{1,2\}$ is a function $f_i:S^*\times S_i\to \Act_i$ such that for $s_0,\ldots,s_n$, we have $f(s_0,\ldots,s_n)\in \Act_i(s_n)$. Thus, a strategy suggests to Player $i$ an available action given the history of the states traversed so far.  
%shaull6
Note that we do not consider randomized strategies, but rather deterministic ones. 
%shaull9
Our results in Section~\ref{sec: solving} show that this is sufficient, in the sense that the players cannot gain by using randomization.

%As we show in Section~\ref{sec: solving}, working with richer  strategies is not needed.
Given strategies $f_1$ and $f_2$ for Players $1$ and $2$, the {\em play} induced $f_1$ and $f_2$ is 
is the infinite sequence of states $s_0,s_1,...$ such that for every $j\ge 0$, if $s_j\in S_i$, for $i \in \{1,2\}$, then $s_{j+1}=\delta_i(s_{j},f(s_0,...,s_{j}))$. For an infinite play $r$, we denote by $\Inf(r)$ the set of states that $r$ visits infinitely often.
The play $r=s_0,s_1,...$ of $\M$ is {\em parity winning} if $\maxs{\alpha(s): s\in\Inf(r)}$ is even. 

The parity-MDP  $\M$ also induces an {\em MDP\/} $\M^{\MDP}=\zug{S_1, S_2,s_0,\Act_1,\Act_2,\delta_1, \delta_2,\MDPProb,\MDPcost}$, obtained by omitting $\alpha$. In this game, we follow the stochastic approach to the environment and consider the distribution of plays when only a strategy for Player~1 is given. Formally, we first extend $\MDPProb$ to transitions as follows: For states $s\in S_2$ and $s'\in S$, we define $\MDPProb(s,s')=\sum_{a\in A(s): \delta_2(s,a)=s'}\MDPProb(s,a)$.
Then, a play of $\M$ with strategy $f_1$ for Player~1 is an infinite sequence of states $s_0,s_1,...$ such that for every $j\ge 0$, if $s_j\in S_1$, then $s_{j+1}=\delta_1(s_{j},f_1(s_0,...,s_{j}))$, and if $s_j\in S_2$, then $\MDPProb(s_j,s_{j+1})>0$. 
The {\em cost} of a strategy $f_1$ is the expected average cost of a random walk in $\M$ in which Player~1 proceeds according to $f_1$. Formally, for $m\in \Nat$ and for a prefix $\tau =s_0,s_1,...s_m$ of a play, let $I_2=\set{j:j<m\text{ and } s_j\in S_2}$. Then, we define $\MDPProb_f(\tau)=\prod_{j\in I} \MDPProb(s_{j},s_{j+1})$ and $\cost_m(f,\tau)=\frac{1}{m+1}\sum_{j=0}^m \MDPcost(s_j)$. 
The cost of a strategy $f_1$ is then  $\cost(f_1)=\liminf_{m\to\infty}\sum_{\tau:|\tau|=m} \cost_m(f_1,\tau)\cdot \MDPProb_f(\tau)$.
%shaull1
We denote by $\Inf(f)$ the random variable that associates $\Inf(\rho)$ with a sequence of states $\rho=s_0,s_1,...$, under the probability space induced by $\M$ with $f$.

A {\em finite memory} strategy for $\M$ is described by a finite set $M$ called {\em memory}, an initial memory $init\in M$, a {\em memory update function} $next:S_1\times M\to M$, and an {\em action function} $act:S_1\times M\to \Act_1$ such that $act(s,m)\in \Act_1(s)$ for every $s\in S_1$ and $m\in M$. 

A strategy is {\em memoryless} if it has finite memory $M$ with $|M|=1$. Note that a memoryless strategy depends only on the current state. Thus, we can describe a memoryless strategy by  $f_1:S_1\to \Act_1$. 
%A memoryless strategy $f$ induces a Markov chain $\M^f=\zug{S,P_f}$ with 
%$P_f(s,s')=\MDPProb(s,f(s),s')$. Let $\pi$ be the limiting distribution of $\M^f$. It is not hard to prove that $cost(f)=\sum_{s\in S}\pi_s \MDPcost(s,f(s))$. 
Let $\cost(\M)=\inf \{\cost(f_1): f_1 \mbox{ is a strategy for }\M\}$. That is, $\cost(\M)$ is the expected cost of a game played on $\M$ in which Player~1 uses an optimal strategy.  

%Theorem~\ref{thm:solving MDP in P} below states that $cost(\M)$ can be attained by a memoryless strategy.
%shaull6
The following is a basic property of MDPs.
\vspace*{-5pt}
\begin{theorem}
	\label{thm:solving MDP in P}
Consider an MDP $\M$. Then, $cost(\M)$ can be attained by a memoryless strategy, which can be computed in polynomial time.
\end{theorem}
\vspace*{-5pt}

Recall that a strategy $f_1$ for player 1 is winning in $M^{\P}$ if every play of $\M$ with $f_1$ satisfies the parity condition $\alpha$. Note that we require {\em sure} winning, in the sense that all plays must be winning, rather than winning with probability 1 ({\em almost-sure} winning). On the other hand, the definition of cost in $M^{\MDP}$ considered strategies for Player~1 and ignore the parity winning condition. We now define the {\em sure cost\/} of the parity-MDP, which does take them into account. 
For a strategy $f_1$ for Player~1, the sure cost of $f_1$, denoted $\costs(f_1)$, is $\cost(f_1)$, if $f_1$ is winning, and is $\infty$ otherwise. The sure cost of $\M$ is then $\costs(\M)=\inf\set{\costs(f_1): f_1\text{ is a strategy for }\M}$.
\vspace*{-8pt}
\paragraph*{End Components}
%shaull7
% and Attractors}
Consider a parity-MDP $\M=\zug{S_1,S_2,s_0,\Act_1,\Act_2,\delta_1, \delta_2,\MDPProb,\MDPcost,\alpha}$. An {\em end component\/} (EC, for short) is a set $U\subseteq S$ such that for every state $s\in U$, the following hold.

%there exists $\emptyset\neq B_s\subseteq A_s$ for which the following hold.
\vspace*{-8pt}
\begin{enumerate}
\item If $s\in S_1$, then there exists an action $a\in \Act_s$ such that $\delta_1(s,a)\in U$.
\item If $s\in S_2$, then for every $a\in \Act_2(s)$ such that $\MDPProb(s,a)>0$, it holds that  $\delta_2(s,a)\in U$.
\item For every $t,t'\in U$, there exist a path $t=t_0,t_1,...,t_k=t'$ and actions $a_1,...,a_t$ such that for every $0\le i< t$, it holds that $t_i\in U$, and there exists an action $a$ such that $\delta(t_i,a)=t_{i+1}$.
\end{enumerate}
Intuitively, the probabilistic player cannot force to leave $U$, and Player~1 has positive probability of reaching every state in $U$ from every other state. 

For an EC $U$ and a state $s\in U$, we can consider the parity-MDP $\M|^s_U$, in which the states are $U$, the initial state is $s$, and all the components are naturally restricted to $U$. Since $U$ is an EC, then this is indeed a parity-MDP. An EC $U$ is {\em maximal\/} if for every nonempty $U' \subseteq S \setminus U$ 
%shaull8
%such that $U' \neq \emptyset$
, we have that $U \cup U'$ is not an EC.

\stam{
Consider a set $R\subseteq S$. A {\em environment attractor} for $R$, denoted $\attr_{\env}(R)$ is defined inductively as follows. First, $T_0=R$. Now, for every $i>0$, let $T_{i+1}=T_i\cup \{s\in S_1: \text{ for every } a\in A_1(s),\text{ we have that }\delta_1(s,a)\in T_i\} \cup \{s\in S_2: \text{ there exists }a\in A_2(s)\text{ s.t. }\delta_2(s,a)\in T_i \text{ and \MDPProb(s,a)>0}\}$. Then, $\attr_{\env}(R)=\bigcup_{i}T_i$. It is well known that $\attr_{\env}(R)$ can be computed in time polynomial in the description of $\M$. We analogously define the {\em system attractor} $\attr_{\sys}(R)$, by swapping the roles of Players $1$ and $2$.
}

\section{Solving Parity MDPs}
\label{sec: solving}
In this section we study the problem of finding the sure cost for an MDP. Recall that for MDPs, there always exists an optimal memoryless strategy. %In fact, the same holds for
%shaull4 
%In addition, finite-memory optimal strategies exist even for attaining the almost-sure cost of parity-MDPs~\cite{CD11}, namely  strategies that win the parity condition w.p. 1 and minimize the expected cost.
We start by demonstrating that for the sure cost of parity-MDPs, the situation is much more complicated.
%\vspace*{-5pt}
\begin{theorem}
\label{no opt}
There is a parity-MDP $\M$ in which Player~1 does not have an optimal strategy (in particular, not a memoryless one) for attaining the sure cost of $\M$. Moreover, 
%shaull9
for every $\epsilon>0$,
Player~1 
%needs 
may need
infinite memory in order to 
%arbitrarily 
$\epsilon$-approximate $\costs(\M)$.
\end{theorem}
%\vspace*{-11pt}
\begin{proof}
Consider the parity-MDP $\M$ appearing in Figure~\ref{fig:infinite mem}. Player~1 can decrease the cost of a play towards $1$ by staying in the initial state. However, in order to ensure an even parity rank, Player~1 must either play $b$ and reach a states with parity rank $2$ and cost $10$ w.p. $0.5$, or play $c$ but incur cost $10$. A finite memory strategy for Player~1 must eventually play $c$ from the initial state in every play,\footnote{Note that this also implies that randomized strategies could not be of help here.} thus the cost of every winning finite-memory strategy is $10$. On the other hand, for every $\epsilon>0$, there exists an infinite memory strategy $f$ that gets cost at most $1+\epsilon$.
Essentially (see Lemma~\ref{lem: GEC value} for a formal proof of the general case), the strategy $f$ plays $b$ for a long time. If the state with parity rank $2$ is reached, it plays $b$ for even longer, and otherwise plays $c$.
\begin{figure}[ht]
\input{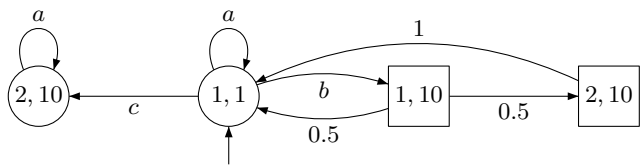}
\caption{The Parity MDP $\M$. States of Player~1 are circles, these of Player~2 are squares, with outgoing edges marked by their probability. Each state is labeled by its parity rank (left) and cost (right). Player~1 has no optimal strategy and needs infinite memory for an $\epsilon$ approximation. }
\label{fig:infinite mem}
\end{figure}

Finally, there is no optimal strategy for Player~1, as every strategy that plays $c$ from the initial state eventually (i.e., as a response to some strategy of Player~2) gets cost $10$ with some positive probability. However, a strategy that never plays $c$ is not parity-winning.
\end{proof}

Following Theorem~\ref{no opt}, our solution to parity MDPs suggests two algorithms. The first, described in Section~\ref{sec:infinite memory}, finds the limit value of all possible strategies, which corresponds to infinite-state transducers. The second, described in Section~\ref{sec:finite memory}, computes the limit value over all finite-memory strategies. 
The complexity of both algorithms is NP$\cap$coNP. Moreover, they are computable in polynomial time when an oracle to a two-player parity game is given.
Hence, our complexity upper bounds match the trivial lower bounds that arise from the fact that every solution to a parity-MDP is also a solution to a parity game.

%In Section~\ref{subsec:overview_discussion} we describe the key differences between the general solution and the solution for finite-memory strategies, and compare our solution to the solution of~\cite{BFRR14a}.

\vspace*{-8pt}
\subsection{Infinite-Memory Strategies}
\label{sec:infinite memory}
In this section we study the problem of finding the sure cost of a parity-MDP when infinite-memory strategies are allowed. 
We prove the upper bound in the following theorem. As stated above, the lower bound is trivial.

\begin{theorem}
\label{thm:parityMDP infinite NP cap co-NP}
Consider a parity-MDP $\M$. Then, $\costs(\M)$ can be computed in NP$\cap$co-NP, and is parity-games hard.
\end{theorem}

Consider a parity-MDP $\M=\zug{S_1,S_2,s_0,\Act_1,\Act_2,\delta_1, \delta_2,\MDPProb,\MDPcost,\alpha}$.
We first remove from $\M$ all states that are not sure-winning for Player~1 in $\M^{\P}$. Clearly, every strategy that attains $\costs(\M)$ cannot visit a state that is losing in  $\M^{\P}$. Thus, we henceforth assume that all states in $\M$ are winning for Player~1 in $\M^{\P}$. 
%In order to compute $\costs(\M)$, 
%orna6 F^sure is not used later.
%This means that there exists a (memoryless) strategy $f^{\rm sure}$ such that for every state $s\in S$, playing $f^{\rm sure}$ from $s$ is winning in $\M^{\P}$.
We say that an EC $C$ of $\M$ is {\em good} (\gec, for short) if its maximal rank is even. That is, $\max_{s\in C}\set{\alpha(s)}$ is even. 

The idea behind our algorithm is as follows. W.p. 1, each play in $\M$ eventually reaches and visits infinitely often all states of some EC. Hence, when restricting attention to plays that are winning for Player~1 in $\M^{\P}$, it must be the case that this EC is good. It follows that the sure cost of $\M$ is affected only by the properties of its \gecs.
Moreover, since  the minimal expected mean-payoff value is the same in all the states of an EC, we can consider only maximal \gecs and refer to the value of an EC, namely the minimal expected value that Player~1 can ensure while staying in the EC. Our algorithm constructs a new MDP (without ranks) $\M'$ in which the cost of a state is the value of the maximal \gec it belongs to. If a state does not belong to a \gec, then we assign it a very high cost in $\M'$, where the intuition is that Player~1 cannot benefit from visiting this state infinitely often. We claim that the sure cost in the parity-MDP $\M$  coincides with the cost of the MDP $\M'$.

Formally, for an EC $C$, let $C^{\max}$ be the set of the states of $C$ with the maximal parity rank in $C$. By definition, this rank is even when $C$ is a \gec. 
%For a play $\rho$ of $\M$, if $\Inf(\rho)=C$, then $\rho$ satisfies the parity condition $\alpha$ and $C$ is a \gec.
Note that if $C$ and $C'$ are \gecs and $C\cap C'\neq \emptyset$, then $C\cup C'$ is also a \gec. Thus, we can restrict attention to maximal \gec.
For a \gec $C$, there exists a memoryless strategy $f^C$ that maximizes the probability of reaching $C^{\max}$ from every state $s\in C$ while staying in $C$. Moreover, since $C$ is an EC, the probability of reaching $C^{\max}$ by playing $f^C$ is strictly positive from every state $s\in C$. Let $t$ be a state in $C$. Consider the MDP $\M^{\MDP}|^t_C$. 
Since $C$ is EC, we have that $\cost(\M^{\MDP}|^t_C)$ is independent of the initial state $t$. Thus, we can define $\cost(\M^{\MDP}|_C)$ as $\cost(\M^{\MDP}|_C^t)$ for some $t\in C$.

Recall that our algorithm starts by a preprocessing step that removes all states that are not sure-winning for Player~1 in $\M^{\P}$.
It then finds the maximal \gecs of $\M$ (using a polynomial-time procedure that we describe in Appendix~\ref{app inf memory finding gecs}), and obtain an MDP $\M'$ by assigning every state within a \gec $C$ the cost $\cost(\M^{\MDP}|_C)$, and assigning every state that is not inside a \gec cost $W+1$, where $W$ is the maximal cost that appears in $\M$. We claim that $\costs(\M)=\cost(\M')$.
 
Before proving the claim, note that all the steps of the algorithm except for the preprocessing step that involves a solution of parity game require polynomial time. In particular, the strategies $f^C$ above are computable in polynomial time by solving a reachability MDP, and, by Theorem~\ref{thm:solving MDP in P}, so does the final step of finding $\cost(M')$.

Proving that $\costs(\M)=\cost(\M')$ involves the following steps (see Appendix~\ref{app inf memory} for the full proof).
%shaull6
First, proving $\costs(\M)\ge \cost(\M')$ is not hard, as a play with a winning strategy $f$ for Player~1 in $\M$ reaches and stays in some \gec $C$ w.p. 1, and within $C$, the best expected cost one can hope for is $\cost(\M^{\MDP}|_C)$, which is exactly what the strategy $f$ attains when played in $\M'$.

Next, proving $\costs(\M)\le \cost(\M')$, we show how an optimal strategy $f'$ in $\M'$ induces an $\epsilon$-optimal strategy $f$ in $\M$. We start with Lemma~\ref{lem: GEC value}, 
which 
 %First, Lemma~\ref{lem: GEC value} 
justifies the 
%assignment of 
costs within a \gec.
%\vspace*{-5pt}
\begin{lemma}
\label{lem: GEC value}
Consider a \gec $C$ in $\M$, and $s\in C$. Let $v(s)=\cost(\M^{\MDP}|^s_C)$, then for every $\epsilon>0$ there exists a strategy $f$ of
%shaull7
%$\M|^s_C$
 $\M^s$
 with $\costs(f)\le v(s)+\epsilon$. 
\end{lemma}
Intuitively, in a good EC, $f$ minimizes the expected mean-payoff and \emph{once in a while} it plays reachability, aiming to visit to a state with the maximal rank in the EC. Since the EC is good, this rank is even.
If reachability is not obtained after $N$ rounds, for a parameter $N$, then $f$ gives up and aims at only surely satisfying the parity objective (our preprocessing step ensures that this is possible).
Otherwise, after reaching the maximal rank, $f$ switches again to minimizing mean-payoff.
This process is repeated forever, increasing $N$ in each iteration.
Hence, the probability that Player~1 eventually gives up can be bounded from above by an arbitrarily small $\epsilon > 0$. Accordingly, Player~1 can achieve a value that is arbitrarily close to $\cost(\M^\MDP|_C)$.

Finally, we construct the $\epsilon$-optimal strategy $f$ in $\M$ as follows. 
%Intuitively, 
The strategy $f$ first mimics $f'$ for a large number of steps $k$, or until an EC (in which $f'$ stays forever) is reached.
If a good EC is not reached, then $f$ aims at only surely satisfying the parity objective. %, Our preprocessing step ensures that this is possible.
If a good EC is reached, then $f$ behaves as prescribed above, per Lemma~\ref{lem: GEC value}. Since the probability of $f'$ reaching a good EC within $k$ steps tends to $1$, then Player~1 can achieve a value within $\epsilon$ of $\cost(\M')$.

\vspace*{-8pt}
\subsection{Finite-Memory Strategies}
\label{sec:finite memory}
In this section we study the problem of finding the sure cost of a parity-MDP, when restricted to finite memory strategies. For a parity-MDP $\M$, we define $\costsf(\M)=\inf\{\costs(f): f $ is a finite memory strategy for $\M\}$. We prove the upper bound in the following theorem. As stated above, the lower bound is trivial.

\begin{theorem}
\label{thm:parityMDP finite NP cap co-NP}
Consider a parity-MDP $\M$. Then, $\costsf(\M)$ can be computed in NP$\cap$co-NP, and is parity-games hard.
\end{theorem}

\vspace*{-5pt}
The general approach is similar to the one we took in Section~\ref{sec:infinite memory}. That is, we remove from $\M$ all states that are not sure-winning for Player 1 in $\M^\P$, and proceed by reasoning about a certain type of ECs. However, for finite-memory strategies, we need a more restricted class of ECs than the \gecs that were used in Section~\ref{sec:infinite memory}. Indeed, a finite-memory strategy might not suffice to win the sure-parity condition in a \gec.

For a \gec $C$, let $k$ be the maximal odd priority in $C$, with $k=-1$ if there are no odd priorities. We define $C^\me=\set{q\in C: \alpha(q)>k\text{ and }\alpha(q)\text{ is even}}$.
 We say that a \gec $C$ in $\M$ is {\em super good} (\sgec, for short) if from every state $s\in C$, there exists a finite-memory strategy $f$ for $\M|_C^s$ such that the play of $\M$ under $f$ reaches $C^{\me}$ w.p. 1, and 
if the play does not reach $C^{\me}$, then it is parity winning.
%  We say that a \gec $C$ in $\M$ is {\em super good} (\sgec, for short) if from every state $s\in C$, there exists a finite-memory strategy $f$ for $\M|_C^s$ such that the following hold.
%\begin{enumerate}
%\item The play of $\M$ under $f$ reaches $C^{\me}$ w.p. 1.
%\item If the play does not reach $C^{\me}$, then it is winning in the parity condition.
%\end{enumerate}
We refer to $f$ as a {\em witness} to $C$ being a \sgec. If $C$ is not a \sgec, we refer to the states of $C$ that satisfy the above as {\em super-good states}.

We argue that \sgecs are the proper notion for reasoning about finite-memory strategies. Specifically, we show that in a \sgec, Player~1 can achieve $\epsilon$-optimal expected cost with a finite-memory strategy, and that every finite-memory winning strategy reaches a \sgec w.p. 1.

Our algorithm finds the maximal \sgecs of $\M$ and obtain an MDP $\M'$ in the same manner we did in Section~\ref{sec:infinite memory}, namely by assigning high weights to states not in \sgecs, and the optimal mean-payoff MDP value to states in \sgecs. As there, we claim that $\cost(\M')=\costsf(\M)$. The analysis of the algorithm as well as its concrete details, are, however, more intricate. 

We start by proving that the notion of maximal \sgecs is well defined. 
To this end, we present the following lemma, whose proof appears in Appendix~\ref{apx:existence of max SGEC}. Note that in the case of \gecs, the lemma was trivial.
\vspace*{-5pt}
\begin{lemma}
\label{lem:existence of max SGEC}
Consider \sgec $C$ and $D$, such that $C\cap D\neq \emptyset$, then $C\cup D$ is also a \sgec.
\end{lemma}
\vspace*{-5pt}
%shaull9
Intuitively, we prove this by considering witnesses $f,g$ for $C$ and $D$ being \sgecs. We then modify $f$ such that from every state in $C$, it tries to reach $D$ for $N$ steps, for some parameter $N$. Once $D$ is reached, $g$ takes over. If $D$ is not reached, $f$ attempts to reach $C^{\me}$. Thus, w.p. 1, the strategy reaches $D^\me$, and if it does not, it either reaches $C^\me$ infinitely often, or wins the parity condition.

%\begin{proof}
%Assume w.l.o.g that the maximal odd priority in $C$ is at least that of $D$. Let $f,g$ be witnesses to $C$ and $D$ being \sgec, respectively. We construct a witness $h$ to $C\cup D$ being a \sgec. In every state $q\in C$, $h$ behaves as $f$ does. In a state $s\in D\setminus C$, $h$ proceeds as follows. (1) It attempts to reach a state $s'\in C\cap D$ (from which it behaves as $f$) within $n_0\in \Nat$ steps, for a large enough $n_0$ such that the probability of reaching $C$ is positive (which exists, since $C\cup D$ is an EC). (2) If $C$ was not reached within $n_0$ steps, $h$ plays $g$ until $C^\me$ is reached, and goes back to (1). 
%
%Observe that $C^\me\subseteq (C\cup D)^\me$. Clearly, the play under $h$ reaches $C^\me$ w.p. 1. Moreover, if the play does not reach $C^\me$, then it is winning in the parity condition. Indeed, if the play under $h$ reaches $C$, then this holds (since $f$ is a witness for $C$ being a \sgec). Otherwise, the play of $h$ either reaches $D^\me$ infinitely often, in which case it is winning in the parity condition, or it plays as $g$ and does not reach $D^\me$, in which case it is parity winning, since $g$ is a witness for $D$ being a \sgec. 
%We conclude that $C\cup D$ is a \sgec.
%\end{proof}
Next, we note that unlike the syntactic definition of \gecs, the definition of \sgecs is semantic, as it involves a strategy. Thus, finding the maximal \sgecs adds another complication to the algorithm. In fact, it is not hard to see that even checking whether an EC is a \sgec is parity-games hard. Using techniques from~\cite{CD11}, we show in Appendix~\ref{apx:verify SGEC} that we can reduce the latter to the problem of solving a parity-\buchi game. 
We thus have the following lemma.
\vspace*{-5pt}
\begin{lemma}
\label{lem:verify SGEC}
Consider an EC $C$ in a parity-MDP $\M$. We can decide whether $C$ is a \sgec in NP$\cap$ co-NP, as well as compute a witness strategy and, if $C$ is not a \sgec, find the set of super-good states.
\end{lemma}
\vspace*{-5pt}

Next, we show how to find the maximal \sgecs of $\M$. Essentially, for every odd rank $k$, we can find the \sgecs whose maximal odd rank is $k$ by removing all states with higher odd ranks, and recursively refining ECs by keeping only super-good states, using Lemma~\ref{lem:verify SGEC}. Thus, we have the following (see Appendix~\ref{apx:maximal SGEC decomposition} for complete details).
\vspace*{-5pt}
\begin{theorem}
\label{thm:maximal SGEC decomposition}
Consider a parity-MDP $\M$. We can find the maximal \sgecs of $\M$ in NP$\cap$co-NP.
\end{theorem}
\vspace*{-5pt}

Theorem~\ref{thm:maximal SGEC decomposition} shows that our algorithm for computing $\costsf(\M)$ solves the problem in NP$\cap$co-NP. 
It remains to prove its correctness.
First, Lemma~\ref{lem: SGEC value} justifies the assignment of costs within a \sgec.
\vspace*{-5pt}
\begin{lemma}
\label{lem: SGEC value}
Consider a \sgec $C$ in $\M$ and a state $s$ in $C$. Let $v(s)=\cost(\M^{\MDP}|^s_C)$. 
Then, for every $\epsilon>0$, there exists a finite-memory strategy $f$ of $\M|^s_C$ with $\costs(f)\le v(s)+\epsilon$. 
\end{lemma}
\vspace*{-10pt}
\begin{proof}
Let $g$ be a memoryless strategy such that $\cost(g)=\cost(\M^{\MDP}|_C^s)$. By Theorem~\ref{thm:solving MDP in P} such a strategy exists. Let $h$ be a finite-memory strategy that witnesses $C$ being a \sgec.
For every $k\in \Nat$, consider the strategy $f_k$ that repeatedly plays $g$ for $k$ steps and then plays $h$ until $C^{\me}$ is reached.
%For every $k\in \Nat$, we construct the strategy $f_k$ as follows. 
%\begin{enumerate}
%\item Play $g$ for $k$ steps.
%\item Play $h$ until $C^{\me}$ is reached.
%\item If $C^{\me}$ is reached, return to Step $1$.
%\end{enumerate}
Since $g$ and $h$ are finite-memory, then $f_k$ is finite memory. In addition, observe that $h$ reaches $C^{\me}$ w.p. 1, and if $C^{\me}$ is not reached, then $h$ is parity-winning. Thus $f_k$ is parity-winning, and it reaches Step 1 infinitely often w.p. 1. Moreover, since $h$ has finite memory, then for every $n\in \Nat$, there is a bounded probability $0<p(n)\le 1$ that $f$ reaches $C^{\me}$ within $n$ steps, with $\lim_{n\to \infty}p(n)=1$. Thus, we get that $\lim_{k\to \infty}\scost(f_k)=\scost(g)=\cost(\M^{\MDP}|^s_C)$, which concludes the proof.
\end{proof}

\vspace*{-10pt}
%shaull6
Lemma~\ref{lem: SGEC value} implies that we can approximate the optimal value of \sgecs with finite-memory strategies. 
It remains to show that it is indeed enough to consider \sgecs. Consider a finite-memory strategy $f$. Then, w.p. 1, $f$ reaches an EC. Let $C$ be an EC with $\Pr_{\M}(\Inf(f)=C)>0$. The following lemma characterizes an assumption we can make on the behavior of $f$ in such an EC.
\vspace*{-5pt}
\begin{lemma}
\label{lem:finite memory EC prob 1}
Consider a parity-MDP $\M$ and an EC $C$. For every finite-memory strategy $f$, if $\Pr_{\M}(\Inf(f)=C)>0$, then there exists a finite-memory strategy $g$ such that for every $s\in C$, we have that $\Pr_{\M^s}(\Inf(g)=C)=1$ and every play of $g$ from $s$ stays in $C$. Moreover, if $f$ is parity winning, then so is $g$.
\end{lemma}
\vspace*{-5pt}
Intuitively, we show that there exists some finite history $h$ such that the strategy $f_h$, which is $f$ played after seeing the history $h$, has the following property: $f_h$ reaches and stays in $C$, and w.p. 1 visits infinitely often all the states in $C$, and in particular $C^{\me}$. 
For the proof, we consider the set $F=\set{f_h:h\text{ is a finite history}}$. Since $f$ has finite memory, it follows that this set is finite. Using this, we show that if $\Pr_\M(\Inf(g)=C)<1$ for every $g\in F$, then $\Pr_\M(\Inf(f)=C)=0$, which is a contradiction. 
Finally, since $f$ is also parity winning, it follows that $f_h$ above is also parity-winning, and is thus a witness for $C$ being a \sgec. 
The full proof appears in Appendix~\ref{apx:finite memory EC prob 1}.

Finally, by Lemma~\ref{lem:finite memory stays in SGEC}, we can assume that once $f$ reaches an EC $C$, it stays in $C$ and visits all its states infinitely often w.p. 1. Since $f$ is parity-winning, it follows that $C$ has a maximal even rank, and that $f$ reaches $C^{\me}$ w.p. 1. Moreover, in every play that does not reach $C^{\me}$, $f$ wins the parity condition. We can thus conclude with the following Lemma, which completes the correctness proof of our algorithm for computing $\costsf(\M)$. See Appendix~\ref{apx:finite memory stays in SGEC} for the proof.
\vspace*{-5pt}
\begin{lemma}
\label{lem:finite memory stays in SGEC}
Consider a parity-MDP $\M$ and an EC $C$. For every finite-memory strategy $f$, if $f$ is parity winning and $\Pr_{\M}(\Inf(f)=C)>0$, then $C$ is a \sgec.
\end{lemma}
%\begin{proof}
%Let $g$ be a strategy obtained as per Lemma~\ref{lem:finite memory EC prob 1}. Thus, $\Pr_{\M^s}(\Inf(g)=C)=1$ for every $s\in C$,  every play of $g$ from $s$ stays in $C$, and $g$ is parity winning. We show that $C$ is a \sgec by showing that $g$ is a witness thereof.
%Indeed, w.p. 1 $g$ visits every state of $C$, and in particular $g$ reaches $C^{\me}$ w.p. 1. In addition, $g$ is parity-winning, so every play of $g$ is parity winning, in particular plays that do not reach $C^{\me}$.
%\end{proof}

\vspace*{-20pt}
\subsection{Comparison with Related Work}
\label{subsec:overview_discussion}
%\emph{Comparison with~\cite{BFRR14a} and~\cite{CR15}.}
Both our work and~\cite{BFRR14a,CR15} analyze ECs and reduce the problem to reasoning about an MDP that ignores the hard constraints.
The main difference with~\cite{BFRR14a} is that there, the hard and soft constraints have the same objective (i.e., worst-case mean-payoff value and expected-case mean-payoff value).
In~\cite{BFRR14a}, the strategy played for $N$ rounds to satisfy the soft objective and then at most $M$ rounds to satisfy the hard objective, for some constants $N$ and $M$.
In our setting, we cannot bound $M$, and in fact it might be the case that Player~1 would play to satisfy the parity objective for the rest of the game (i.e., forever) even after reaching a super-good end component.

The difference with~\cite{CR15} is twofold. First, technically, the type of hard constraints in~\cite{CR15} is worst-case mean-payoff, whereas our setting uses the Boolean parity condition. In classical parity games, the parity condition can be reduced to a mean-payoff objective. Similar reductions, however, do not work in order to reduce our setting to the setting of~\cite{CR15}. Thus, our contribution is orthogonal to~\cite{CR15}. Secondly, Boolean constraints are conceptually different than quantitative constraints, and as we demonstrate in Section~\ref{sec:applications}, they arise naturally in quantitative extensions of Boolean paradigms.

We note that~\cite{CR15} also study a relaxation in which  almost-sure winning is allowed for the hard constraints. An analogue in our setting is to consider an almost-sure parity condition. We note that in such a setting, \gecs are sufficient for reasoning both about finite-memory and infinite-memory strategies. 
%shaull9
Moreover, the preprocessing involves solving an almost-sure parity MDP (without mean-payoff constraints), which can be done in polynomial time.
Thus, as is the case in~\cite{CR15}, we can compute the cost of an MDP with almost-sure hard constraints in polynomial time.
%The works differ in the way they analyze the end-components.
%The main difference arises from the fact that in~\cite{BFRR14a,CR15} the hard and soft constraints have the same type, namely mean-payoff, whereas 
%same objective (e.g., worst-case mean-payoff value and expected-case mean-payoff value).
%In our work, the objectives are orthogonal, and thus playing to satisfy the hard objective might affect the soft objective and vice-versa.
%Thus, a different approach is needed to analyze the end-components.
%The flavor of the resulting finite-memory strategies is also different.
%In~\cite{BFRR14a} the strategy played for $N$ rounds to satisfy the soft objective and then at most $M$ rounds to satisfy the hard objective, for some constants $N$ and $M$.
%In our strategy, we cannot bound $M$, and in fact it might be the case that the player will play to satisfy the parity objective for the rest of the game (i.e., forever), even after reaching a super-good component.

\section{Applications}
\label{sec:applications}

In this section we study two applications of parity-MDPs. Both extend the Boolean synthesis problem. 
Due to lack of space, our description is only an overview. Full definitions and details can be found in Appendix~\ref{app app}.
%shaull9 
We start with some basic definitions.

For finite sets $I$ and $O$ of input and output signals, respectively, an {\em $I/O$  transducer} is $\T=\zug{I,O,Q,q_0,\delta,\rho}$, where $Q$ is a set of states, $q_0 \in Q$ is an initial state, $\delta: Q\times \tIN \to Q$ is a total (deterministic) transition function, and $\rho:Q\to \tOUT$ is a labeling function on the states. The run of $\T$ on a word $w=i_0 \cdot i_1 \cdots \in \tINo$ is the sequence of states $q_0,q_1,\ldots$ such that $q_{k+1} = \delta(q_k,i_{k})$ for all $k \geq 0$. The {\em output} of $\T$ on $w$ is then $o_1,o_2,\ldots\in \tOUTo$ where $o_k=\rho(q_{k})$ for all $k\ge 1$. Note that the first output assignment is that of $q_1$, and we do not consider $\rho(q_0)$. This reflects the fact that the environment initiates the interaction. The {\em computation of $\T$ on $w$\/} is then 
$\T(w)=i_0\cup o_1,i_1\cup o_2,\ldots \in (2^{I \cup O})^\omega$.
When $Q$ is a finite set, we say that the transducer is  finite.

The synthesis problem gets as input a specification $L \subseteq (2^{I \cup O})^\omega$ and generates a transducer $\T$ that realizes $L$; namely, all the computations of $\T$ are in $L$. The language $L$ is typically given by an LTL formula \cite{Pnu81} or by means of an automaton of infinite words.

\vspace*{-8pt}
\subsection{Penalties on Undesired Scenarios}
Recall that in the Boolean synthesis problem, the goal is to generate a transducer that associates with each infinite sequence of inputs an infinite sequence of outputs so that the result computation satisfies a given specification. Typically, some behaviors generated by the transducers may be less desired than others. For example, as discussed in Section~\ref{sec:intro}, designs that use fewer resources or minimize expensive activities are preferable. The input to the {\em synthesis with penalties\/} problem includes, in addition to the Boolean specification, languages of finite words that describe undesired behaviors, and their costs. 
The goal is to generate a transducer that realizes the specification and minimizes cost due to undesired behaviors. 

Formally, the input to the problem includes languages $L_1,\ldots,L_m$ of finite words over the alphabet $2^{I \cup O}$ and a penalty function $\gamma:\{1,\ldots,m\}\to \Nat$ specifying for each $1 \leq i \leq m$ the penalty that should be applied for generating a behavior in $L_i$. As described in Section~\ref{sec:intro}, the language $L_i$ may be local (that is, include only words of length 1) and thus refer only to activation of output signals, may specify short scenarios like flips of output signals, and may also specify rich regular scenarios. Note that we allow penalties also for behaviors that depend on the input signals. Intuitively, whenever a computation $\pi$ includes a behavior in $L_i$, a penalty of $\gamma(i)$ is applied. Formally, if $\pi=\sigma_1,\sigma_2,\ldots$, then for every position $j \geq 1$, we define ${\it penalty}(j) = \{i : \mbox{ there is } k \leq i \mbox{ such that } \sigma_k \cdot \sigma_{k+1} \cdots \sigma_j \in L_i\}$. That is, ${\it penalty}(j)$ points to the subset of languages $L_i$ such that a word in $L_i$ ends in position $j$. Then, the cost of position $j$, denoted $\cost(j)$, is $\sum_{i \in {\it penalty}(j)} \gamma(i)$. 
Finally, for a finite computation $\pi=\sigma_1,\sigma_2,\ldots$, we define its cost, denoted $\cost(\pi)$, as $\limsup_{m \rightarrow \infty}\frac{1}{m}\sum_{j=1}^{m}{\cost(j)}$.

Let $\A$ be a deterministic parity automaton (DPW, for short) over the alphabet $2^{I \cup O}$ that specifies the specification $\psi$. 
We describe a parity-MDP whose solution is a transducer that realizes $\A$ with the minimal cost for penalties. The idea is simple: on top of the parity game $\G$ described above, we compose monitors that detect undesired scenarios. We assume that the languages $L_1,\ldots,L_m$ and are given by means of deterministic automata on finite words (DFWs)  $\U_1,\ldots,\U_m$ where for every $1\le i\le m$, we have that $L(\U_i)=(2^{I\cup O})^*\cdot L_i$. That is, $\U_i$ accepts $\sigma_1\cdots \sigma_n$ iff there exists $k\le n$ such that $\sigma_k\cdots \sigma_n\in L_i$. Essentially, we turn $\A$ into a parity-MDP by composing it with the DFWs $\U_1,\ldots,\U_m$. Then, $\U_i$ reaching an accepting state indicates that the penalty for $L_i$ should be applied, which induces the costs in the parity-MDP. The probabilities in the parity-MDP are induced form the distribution of the assignments to input signals. The full construction can be found in Appendix~\ref{app penalties}. We note that an alternative definition can replace the DFWs $\U_1,\ldots,\U_m$ and the cost function $\gamma$ by a single weighted automaton that can be composed with $\A$.

\subsection{Sensing}
\label{sec:sensing}
%In \cite{AKK14}, regular sensing is defined as a measure for the number of sensors that need to be operated in order to recognize a regular language.
Consider a transducer  $\T=\zug{I,O,Q,q_0,\delta,\rho}$. For a state $q\in Q$ and a signal $p\in I$, we say that $p$ is {\em sensed in\/} $q$ if there exists a set $S\subseteq I$ such that $\delta(q,S  \setminus\set{p})\neq \delta(q,S\cup \set{p})$. Intuitively, a signal is sensed in $q$ if knowing its value may affect the destination of at least one transition from $q$. We use $\sen(q)$ to denote the set of signals sensed in $q$. 
The {\em sensing cost\/} of a state $q\in Q$ is $\stcost(q)=|\sen(q)|$. 
%shaull5
%\footnote{We note that, alternatively, one could define the {\em sensing level\/} of states, with $\slevel(q)=\frac{\stcost(q)}{|P|}$. Then, for all states $q$, we have that $\slevel(q) \in [0,1]$. All our results hold also for this definition, simply by dividing the sensing cost by $|P|$.}
For a finite run $r=q_1,\ldots,q_m$ of $\T$, we define the sensing cost of $r$, denoted $\scost( r)$, as $\frac{1}{m}\sum_{i=0}^{m-1}{\stcost(q_i)}$. That is, $\scost(r)$ is the average number of sensors that $\T$ uses during $r$. For a finite input sequence $w \in (2^I)^*$, we define the sensing cost of $w$ in $\T$, denoted $\scost_{\T}(w)$, as the sensing cost of the run of $\T$ on $w$. 
Finally, the sensing cost of $\T$ is the expected sensing cost of input sequences of length that tends to infinity, which is parameterized by a distribution 
on $(\tIN)^\omega$ given by a sequence of distributions $D_1,D_2,...$ such that $D_t:\tIN\to [0,1]$ describes the distribution over $\tIN$ at time $t\in \Nat$. For simplicity, we assume that the distribution is uniform. Thus, $D_t(i)=2^{-|I|}$ for every $t\in \Nat$. 
For the uniform distribution we have $\scost(\T)=\lim_{m\to \infty} |(2^I)|^{-m}\sum_{w \in (2^I)^m} \scost_{\T}(w)$. 

Note that this definition also applies when the transducer is infinite. However, for infinite transducers, the limit in the definition of $\scost(\T)$ might not exist, and we therefore define $\scost(\T)=\limsup_{m\to \infty}|\tIN|^{-m}\sum_{w\in (\tIN)^m}\scost_\T(w)$. Finally, for a realizable specification $L\in 2^{I\cup O}$, we define $\scost_{I/O}(L)=\inf\{\scost(\T): \T $ is an $I/O$ transducer that realizes $L\}$.

In~\cite{AKK15}, we study the sensing cost of {\em safety} properties. 
%shaull9
We show that there, a finite, minimally-sensing transducer, always exists (albeit of exponential size), and the problem of computing the sensing cost is EXPTIME-complete. In our current setting, however, a minimally-sensing transducer need not exist, and 
%shaull9
%an $\epsilon$-approximation 
any approximation may require infinite memory. 
We demonstrate this with an example. 

\begin{example}
Let $I=\set{a}$ and $O=\set{b}$, and consider the specification $\psi=(\Alw \Ev a\wedge \Alw b)\vee \Alw(\neg b\to \Next\Alw(a\leftrightarrow b))$. Thus, $\psi$ states that either $a$ holds infinitely often and $b$ always holds, or, if $b$ does not holds at a certain time, then henceforth, $a$ holds iff $b$ holds. Observe that once the system outputs $\neg b$, it has to always sense $a$ in order to determine the output. 
The system thus has an incentive to always output $b$. This, however, may render $\psi$ false, as $a$ need not hold infinitely often.

We start by claiming that every finite-memory transducer $\T$ that realizes $\psi$ has sensing cost 1. Indeed, let $n$ be the number of states in $\T$. A random input sequence contains the infix $(\neg a)^{n+1}$ w.p. 1. Upon reading such an infix, $\T$ has to output $\neg b$, as otherwise it would not realize $\psi$ on a computation with suffix $(\neg a)^\omega$. Thus, from then on, $\T$ 
%shaull6
%has to sense 
senses
$a$ in every state. So $\scost(\T)=1$.

%shaull6
However, by using infinite-memory transducers, we can follow the construction in Section~\ref{sec:infinite memory} and reduce the sensing cost arbitrarily close to $0$. Let $M\in \Nat$. We construct a transducer $\T'$ as follows. After initializing $i$ to $1$, the transducer $\T'$ senses $a$ and outputs $b$ for $iM$ steps. If $a$ does not hold during this time, then $\T'$ outputs $\neg b$ and starts sensing $a$ and outputting $b$ accordingly. Otherwise, if $a$ holds during this time, then $\T'$ stops sensing $a$ for ${2^{i}}$ steps, while outputting $b$. It then increases $i$ by $1$ and repeats the process. 
%shaull6
Note that $\T'$ outputs $\neg b$ iff $a$ does not hold for $iM$ consecutive positions at the $i$-th round (which happens w.p. $2^{-iM}$). Thus, the probability of $\T'$ outputting $\neg b$ in a random computation is bounded from above by $\sum_{i=1}^\infty 2^{-iM}=2^{-M}$, which tends to $0$ as $M$ tends to $\infty$. Note that in the $i$-th round, $\T'$ senses $a$ for only $iM$ steps, and then does not sense anything for $2^{iM}$ steps, so if $\T'$ does not output $\neg b$, the sensing cost is $0$. Thus, we have 
%Similarly to the proof of Lemma~\ref{lem: SGEC value}, we have 
$\lim_{M\to \infty}\scost(\T')=0$. \qed

%Note that the second condition requires sensing 1.

%Any finite memory strategy will have to satisfy condition 2. If the memory size is M, then with probability 1, eventually we will have M + 1 neg a in a row, and it will have to output neg b. So the sensing is 1.
%
%An infinite memory strategy will sense the first M rounds. If no "a" played, then give up and play according to 2.
%Then don't sense for 2^M rounds, and then increase M by 1, and repeat the process.
%The probability of giving up is roughly 2^{M-1}. If we don't give up the sensing is 0.
%So the expected sensing is 2^{M-1} and the limit is 0.

\end{example}

%shaull6
We proceed by describing the general solution to computing the sensing cost of a specification.
Recall that synthesis of a DPW $\A$ is reduced to solving a parity game.
When sensing is introduced, it is not enough for the system to win this game, as it now has to win while minimizing the sensing cost. Intuitively, not sensing some inputs introduces incomplete information to the game: once the system gives up sensing, it may not know the state in which the game is and knows instead only a set of states in which the game may be. In particular, unlike usual realizability, a strategy that minimizes the sensing need not use the state space of the \DPW. 

\vspace*{-5pt}
\begin{theorem}
\label{thm:sensing to parity-MDP}
Consider a $\DPW$ specification $\A$ over $2^{I\cup O}$. There exists a parity-MDP $\M$ such that $\costs(\M)=\scost_{I/O}(L(\A))$.	
Moreover, the number of states of $\M$ is singly exponential in that of $\A$, and the number of parity ranks on $\M$ is polynomial in that of $\A$.
\end{theorem}
\vspace*{-10pt}
\begin{proof}
Conceptually, we follow the ideas of Boolean synthesis, by thinking of $\A$ as a parity game between the system and the environment, as described in Section~\ref{sec:application defs}. The proof is comprised of several steps. First, intuitively, we give the system an option to sense only some input signals $x\subseteq I$, but require that then, the play must be winning for every assignment of the inputs that are not sensed. Then, we introduce costs induced by the number of sensed input signals in each state, and finally we add a uniform stochastic environment. 
Technically, however, the first step is done using automata, rather than games, and converts the $\DPW$ $\A$ into a {\em universal parity automaton\/} (UPW) -- an automaton in which a the transition function maps each state and letter to more than a single successor state, and a word is accepted if all the runs on it are accepting. We use the universal branches of the UPW in order to model the several possible assignments to the input signals that are not sensed. Thus, in state $s$ of $\A$, the system chooses a state $\zug{s,x}$, where $x\subseteq I$ represents the inputs to be sensed. The environment then chooses an assignment $i:I\to \set{0,1}$ for the inputs, and the system chooses an output assignment $o:O\to \set{0,1}$. However, instead of the new state being $\delta(s,i\cup o)$, a universal transition is taken to every state $s'$ such that $s'=\delta(s,i'\cup o)$ for some  $i'$ that agrees with $i$ on every input in $x$. Thus, effectively, the system has to play only according to the values of the sensed inputs. Note that the two players are not modeled in the automaton. Rather, their choices are represented by augmenting the alphabet to include a $2^I$ component to represent the sensed inputs.
Using automata allows us to determinize the $\UPW$ back to a $\DPW$ that already captures sensing. We then convert the automaton to a parity-game, and proceed as described above.

For the formal details, see Appendix~\ref{app sen construction}.
\end{proof}
%of stam
\vspace*{-12pt}
\begin{theorem}
\label{thm:compute sensing cost}
Consider a $\DPW$ specification $\A$ over $2^{I\cup O}$. We can compute $\scost_{I/O}(L(\A))$ in singly-exponential time. Moreover, the problem of deciding whether $\scost_{I/O}(L(\A))>0$ is EXPTIME-complete.
\end{theorem}
\vspace*{-13pt}
\begin{proof}
We obtain from $\A$ a parity-MDP $\M$ as per Theorem~\ref{thm:sensing to parity-MDP}. Observe that the algorithm in the proof of Theorem~\ref{thm:parityMDP infinite NP cap co-NP} essentially runs in polynomial time, apart from solving a parity game, which is done in NP$\cap$co-NP. However, deterministic algorithms for solving parity games run in time polynomial in the number of states, and singly-exponential in the number of parity ranks. Since the number of parity ranks in $\M$ is polynomial in that of $\A$, we can find $\costs(\M)$ in time singly-exponential in the size of $\A$. Since $\costs(\M)=\scost_{I/O}(L(\A))$, we are done.

For the lower bound, we note that the problem of deciding whether $\scost_{I/O}(L(\A))>0$ is EXPTIME-hard even for a restricted class of automata, namely looping automata~\cite{AKK15}. 
\end{proof}

The input to the synthesis problem is typically given as an \LTL formula, rather than a \DPW. Then, the translation from $\LTL$ to a $\DPW$ involves a doubly-exponential blowup. Thus, a naive solution for computing the sensing cost of a specification given by an $\LTL$ formula is in 3EXPTIME. However, by translating the formula to a \UPW, rather than a \DPW, we show how we can avoid one exponent, thus matching the 2EXPTIME complexity of standard Boolean synthesis.
\vspace*{-5pt}
\begin{theorem}
\label{thm:LTL sensing cost}
Consider an \LTL specification $\psi$ over $I\cup O$. We can compute $\scost_{I/O}(L(\psi))$ in doubly-exponential time.
\end{theorem}
\vspace*{-12pt}
\begin{proof}
We start by translating $\psi$ to a $\UPW$ $\A$ of size single-exponential in the size of $\psi$. This can be done, for example, by translating $\neg \psi$ to a nondeterministic \buchi automaton~\cite{VW94} and dualizing it. We then follow the proof of Theorem~\ref{thm:compute sensing cost}, by adding the universal transitions described there directly to the $\UPW$ $\A$. Thus, when we finally determinize the $\UPW$ to a $\DPW$, the size of the $\DPW$ is doubly-exponential, 
%and thus 
so 
computing the sensing cost can also be done in doubly-exponential time.
\end{proof}

\small
\bibliographystyle{plain}
\bibliography{../ok}

\normalsize

\appendix

\section{Calculating the Sure Cost in the Infinite-Memory Case}
\label{app inf memory}

Our algorithm uses the notion of attractors, defined below.
Consider a set $R\subseteq S$. A {\em environment attractor} for $R$, denoted $\attr_{\env}(R)$ is defined inductively as follows. First, $T_0=R$. Now, for every $i>0$, let $T_{i+1}=T_i\cup \{s\in S_1: \text{ for every } a\in A_1(s),\text{ we have that }\delta_1(s,a)\in T_i\} \cup \{s\in S_2: \text{ there exists }a\in A_2(s)\text{ s.t. }\delta_2(s,a)\in T_i \text{ and \MDPProb(s,a)>0}\}$. Then, $\attr_{\env}(R)=\bigcup_{i}T_i$. It is well known that $\attr_{\env}(R)$ can be computed in time polynomial in the description of $\M$. We analogously define the {\em system attractor} $\attr_{\sys}(R)$, by swapping the roles of Players $1$ and $2$.

\subsection{Finding the Maximal \gecs of $\M$}
\label{app inf memory finding gecs}

In order to find the maximal \gecs of $\M$, we proceed as follows. 
\begin{enumerate}
\item Compute the maximal EC decomposition of $\M$.
\item For every maximal EC $C$, if $C$ is not good (i.e., the maximal parity rank in $C$ is odd), remove it from the graph $\attr_{\env}(C^{\max})$ and go to (1).
\item Once all the remaining components are good, return them.
\end{enumerate}
Note that upon returning to step (1) from (2), it may be that the graph of $\M$ is not connected. Still, we find the decomposition in all the components.

It is not hard to see that all the steps of the algorithm are polynimial. In particular, finding the maximal EC decomposition of $\M$ takes polynomial time \cite{CH14}.

%\section{Proofs}
\subsection{Proof of Lemma \ref{lem: GEC value}}

Consider a memoryless strategy $f^C$ that maximizes the probability to reach $C^{\max}$, and a memoryless strategy $g$ whose expected cost in $\M^{\MDP}|_C^s$ is $v(s)=v(C)$. By Theorem~\ref{thm:solving MDP in P}, such a strategy $g$ exists.

We construct an infinite-memory strategy $h$ that works in phases, as follows. In phase 1, $h$ works in iterations. In iteration $i$, the strategy $h$ plays $g$ for $2^{2^i}$ steps. Then, $h$ plays $f^C$ for $\gamma_\epsilon \cdot n\cdot 2^i$ steps, where $\gamma_\epsilon$ is a constant we determine later and $n$ is the number of states of $\M$. If, during these $2^{2^i}+\gamma_\epsilon 2^i$ steps, the generated play reached $C^{\max}$, then we proceed to the next iteration. Otherwise, $h$ goes to phase 2, in which it plays a parity-winning strategy (which exists, since every state in $\M$ is parity winning).

Clearly, if the play generated by $h$ never reaches Phase 2, then playing $g$ for $2^{2^i}$ steps is the dominant factor, and we have that $\cost(h)=\cost(g)\le v(s)$. Thus, it remains to bound the probability that the play reaches Phase 2. 
Denote by $\lambda$ the maximal probability that a play of $f^C$ does not reach $C^{\max}$ within $n$ steps, where the maximum is taken over all states of $C$. Since $C$ is strongly connected, it follows that $0\le \lambda<1$. Thus, the probability of not reaching Phase 2 is bounded from below by $\prod_{i=1}^\infty (1-\lambda^{\gamma_\epsilon 2^i})$. The latter expression converges to a number $p$ in $(0,1]$ that is 
%(as a sub-product of the Euler function). 
%Moreover, $p$ is 
inversely-related to $\gamma_\epsilon$. Therefore, by setting $\gamma_\epsilon$ large enough, we can lower the probability of reaching Phase 2 arbitrarily. Since the cost of a play after reaching Phase 2 is bounded from above by $W$, the claim follows.

\subsection{A proof that $\costs(\M)=\cost(\M') $}

We start with the ``easy'' direction, proving that $\costs(\M)\ge \cost(M')$. Consider a winning strategy $f$ for $\M$. With probability $1$, the play of $f$ in $\M$ reaches and stays in some \gec $C$. From every state in $C$, the minimal expected cost (when staying in $C$) is $v(C)$. Indeed, $v(C)$ is the cost of an MDP without the parity condition, which can only lower the minimal expected cost.
 Thus, we have that $\cost_\M(f)\ge \sum_{C\text{ is a \gec}}\Pr(f \text{ reaches and stays in }C)\cdot v(C)$.

Consider the strategy $f$ as a strategy for $\M'$. Then, $\cost_{\M'}(f)=\sum_{C\text{ is a \gec}}\Pr(f \text{ reaches and stays in }C)\cdot v(C)$, and we conclude that $\costs(\M)\ge \cost(\M')$.

For the other direction, we show that $\costs(\M)\le \cost(\M')$. Since $\M'$ is an MDP, then there exists an optimal memoryless strategy $f'$ such that $\cost_{\M'}(f')=\cost(\M')$. We show that for every $\epsilon>0$, there exists a winning strategy $f$ for $\M$ such that $\cost_{\M}(f)\le \cost_{\M'}(f')+\epsilon$. 

Observe that since $f'$ is memoryless and optimal, there exists a set of ECs ${\C}$  such that for every $C\in \C$, once $f'$ reaches a state $s\in C$, it stays in $C$ forever. Moreover, observe that every $C\in \C$ must be a \gec. Indeed, the states outside a \gec in $\M'$ have value $2W+1$, but from every state in $\M$ there exists a strategy that is parity-winning, and therefore ensures that a \gec is reached. Thus, if $f'$ gets stuck in an EC that is not good, we can modify it to reach a \gec, thus decreasing its cost. 

Let $\epsilon>0$. There exists some $N_0\in \Nat$ such that after $N_0$ steps, w.p. at least $1-\epsilon'$ a play in $\M'_{f'}$ reaches a \gec in $\C$ (for $\epsilon'>0$ which we will fix later). We obtain $f$ from $f'$ as follows. $f$ simulates $f'$ for $N_0$ steps. During this simulation, whenever $f$ reaches a \gec $C\in \C$, $f$ starts playing the strategy described in the proof of Lemma~\ref{lem: GEC value} for $\epsilon'$. After $N_0$ steps, $f$ plays a parity-winning strategy. 

Clearly $f$ is parity-winning. In addition, by our choice of $N_0$ and by Lemma~\ref{lem: GEC value}, it follows that w.p. at least $1-\epsilon'$, the cost of $f$ is at most $\cost_{\M'}(f)+\epsilon'$. Thus, $\cost_{\M}(f)\le (1-\epsilon')(\cost_{\M'}(f)+\epsilon')+\epsilon'|W|$, and for a small enough $\epsilon'$, this is at most $\cost_{\M'}(f)+\epsilon$.

\section{Calculating the Sure Cost in the finite-Memory Case}
\label{apx:finite memory}

\subsection{Proof of Lemma~\ref{lem:existence of max SGEC}}
\label{apx:existence of max SGEC}
Assume w.l.o.g that the maximal odd priority in $C$ is at least that of $D$. Let $f,g$ be witnesses to $C$ and $D$ being \sgec, respectively. We construct a witness $h$ to $C\cup D$ being a \sgec. In every state $q\in C$, $h$ behaves as $f$ does. In a state $s\in D\setminus C$, $h$ proceeds as follows. (1) It attempts to reach a state $s'\in C\cap D$ (from which it behaves as $f$) within $n_0\in \Nat$ steps, for a large enough $n_0$ such that the probability of reaching $C$ is positive (which exists, since $C\cup D$ is an EC). (2) If $C$ was not reached within $n_0$ steps, $h$ plays $g$ until $C^\me$ is reached, and goes back to (1). 

Observe that $C^\me\subseteq (C\cup D)^\me$. Clearly, the play under $h$ reaches $C^\me$ w.p. 1. Moreover, if the play does not reach $C^\me$, then it is winning in the parity condition. Indeed, if the play under $h$ reaches $C$, then this holds (since $f$ is a witness for $C$ being a \sgec). Otherwise, the play of $h$ either reaches $D^\me$ infinitely often, in which case it is winning in the parity condition, or it plays as $g$ and does not reach $D^\me$, in which case it is parity winning, since $g$ is a witness for $D$ being a \sgec. 
We conclude that $C\cup D$ is a \sgec.

\subsection{Proof of Lemma~\ref{lem:verify SGEC}}
\label{apx:verify SGEC}
Our solution proceeds as follows. We start by reducing the problem of deciding whether $C$ is a \sgec to the problem of deciding whether there is a winning strategy in a parity-\buchi game, using techniques from~\cite{CD11}. We then show how the latter can be solved by a reduction to positive Mean-Payoff parity games.

A parity-\buchi game is a two player game $G=\zug{S_1,S_2,s_0,\Act_1,\Act_2,\delta_1, \delta_2,(\alpha,\beta)}$ that is similar to a parity game, with the exception that the winning condition is composed of two conditions: $\alpha$ is a parity ranking function, and $\beta\subseteq Q$ is a set of {\em accepting states}. A play of $G$ is winning for Player~1 iff it satisfies the parity condition $\alpha$, and visits $\beta$ infinitely often.

We start by describing a reduction from the problem of deciding whether $C$ is a \sgec to the problem of solving a parity-\buchi game. First, we check that $C$ is a \gec. If $\max_{s\in C}\set{\alpha(s)}$ is odd, then $C$ is not a \sgec and we are done.

Consider the parity game $\M^\P|_C$. We obtain from $\M^\P|_C$ a parity-\buchi game $G$ as follows. 
First, we change every state in $C^{\me}$ to a \buchi accepting sink (while keeping the parity rank).

For every state $s$ of Player~2 that is not in $C^{\me}$, we replace $s$ with the gadget in Figure~\ref{fig:gadget}.
\begin{figure}[ht]
\input{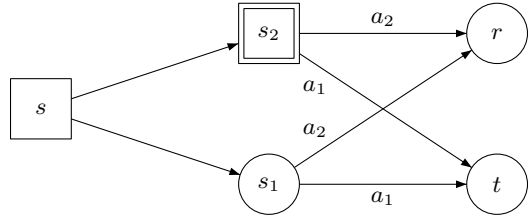}
\caption{Gadget for the reduction in Lemma~\ref{lem:verify SGEC}. In $\M$, we have $\delta(s,a_1)=t$ and $\delta(s,a_2)=r$.}
\label{fig:gadget}
\end{figure}

Formally, we add the states $s_1,s_{2}$, where $s_1$ is a Player~1 state and $s_{2}$ is a Player~2 state, whose successors are those of $s$ (with the same available actions), and the successors of $s$ are $s_1$ and $s_{2}$.

We set the parity ranks of the gadget to be $\alpha(s_1)=\alpha(s_{2})=0$, and for the \buchi objective, we set $s_2\in \beta$ and $s_1\notin \beta$.

We claim that $C$ is a \sgec iff Player~1 wins in $G$ from every state. For the first direction, assume $C$ is a \sgec, and let $f$ be a witness strategy. Thus, $f$ is finite memory strategy that reaches $C^{\me}$ w.p. 1, and wins in the parity condition in every play that does not reach $C^{\me}$.

We obtain from $f$ a strategy $g$ for Player~1 in $G$ as follows. $g$ plays similarly to $f$, unless a state $s_1$ as in the gadget  is reached, for some environment state $s$. Then $g$ chooses the neighbor that minimizes the distance to $\attr_\sys(C^{\me})$ (we assume w.l.o.g that in $\attr_\sys(C^{\me})$, the strategy $f$ leads surely to $C^{\me}$). We claim that $g$ wins parity+\buchi in $G$. Indeed, consider a strategy $g'$ for Player~2 in $G$, and consider the play $\rho$ induced by $g$ and $g'$. Note that $g'$ induces a strategy in $\M^\P|_C$ by assigning each state $s\in S_2$ the action $g(s_2)$.
Assume by way of contradiction that $\rho$ is not winning for Player~1. Thus, either the \buchi condition or the parity conditions do not hold. If the \buchi condition does not hold, then after a finite prefix, for every environment state $s$, $g'$ moves the play to $s_1$ in the gadget (since $s_2\in \beta$). Thus, however, eventually Player~1 forces the play to $C^{\me}$, which are \buchi-winning sinks, and this the \buchi condition and the parity conditions are satisfied. Thus, the \buchi condition holds. If the parity condition does not hold, then the play does not reach $C^{\me}$. Since $f$ is a witness strategy for $C$ being a \sgec, then every play in $\M^\P|_C$ induced by $f$ and does not reach $C^{\me}$ is parity winning. Thus, the play in $G$ induces similar parity ranks, with the exception of padding $0$ ranks within the gadgets. In particular, this play is also parity winning in $G$. Since this is true for every strategy $g'$, we conclude that $f$ is parity-\buchi winning in $G$.

conversely, assume that $f$ is a parity-\buchi winning strategy in $G$. In addition, we assume that $f$ is finite memory. Since parity-\buchi is an $\omega$-regular winning condition, then Player~1 has a finite-memory winning strategy. The strategy $f$ induces a strategy for Player~1 in $\M|_C$. We claim that this is a witness for $C$ being a \sgec.
Indeed, similarly to the above, it is easy to see that $f$ is parity winning if $C^{\me}$ is not reached. It remains to prove that $C^{\me}$ is reached w.p. 1. 

Since $f$ has finite memory, then there exists $n\in \Nat$ such that for every state $s$ in $G$, if Player~2 chooses $t_1$ from every environment state $t$ for $n$ steps, then $f$ reaches $C^{\me}$. However, w.p. 1, a stochastic environment chooses the same $n$ choices that $f$ would have chosen in the above $t_1$ states. Thus, w.p. 1, $f$ reaches $C^{\me}$.

This completes the reduction to parity-\buchi games.

Next, we reduce parity-\buchi games to Mean-payoff parity games by assigning every state in $\beta$ payoff 1, and the rest payoff 0. Then, the goal is to win parity while having strictly positive long-run mean-payoff. These games can be solved in NP$\cap$co-NP~\cite{CD11}.

In addition, in case $C$ is not a \sgec, our solution finds the winning states for Player~1, which are the super-good states.

\subsection{Proof of Theorem~\ref{thm:maximal SGEC decomposition}}
\label{apx:maximal SGEC decomposition}
Intuitively, our algorithm works in two phases. First, for every odd rank $k$, we find the maximal \sgec whose maximal odd rank is $k$. Then, we choose among the \sgec the maximal ones. We start by describing a subroutine for the first phase.

Let $d$ be the maximal parity rank in $\M$, and consider an odd rank $k\in \set{-1,\ldots,d}$. We compute the maximal \sgec with maximal odd rank $k$ as follows.
\begin{enumerate}
\item Compute the maximal EC decomposition $\C$.
\item For every EC $C\in \C$,
\begin{enumerate}
\item Let $odd_{>k}(C)=\set{s\in C: \alpha(s)>k \text{ and }\alpha(s)\text{ is odd}}$. If $odd_{>k}(C)\neq \emptyset$, remove $\attr_{\env}(odd_{>k})$ from $C$, and go to (1).
\item Decide if $C$ is a \sgec. If it is, return it. Otherwise, find the set $W$ of super-good states, remove $\attr_{\env}(C\setminus W)$ from $C$, and go to (1).
\end{enumerate}
\end{enumerate}
Next, we run this subroutine for every odd $k\in \set{-1,\ldots ,d}]$ to obtain \sgec $C_1,...,C_m$. Finally, for every $C_i,C_j$, if $C_i\subseteq C_j$, we remove $C_j$ from the list. 

Clearly this algorithm has polynomially many iterations, and in each iteration we solve an NP$\cap$co-NP problem, as per Lemma~\ref{lem:verify SGEC}. Thus, the algorithm solves the problem in NP$\cap$co-NP.

It remains to prove the correctness of the algorithm.
By Lemma~\ref{lem:verify SGEC}, every component that is returned in the subroutine is a \sgec. Consider a \sgec $C$ with maximal odd rank $k$. In iteration $k$ of the algorithm, none of the states of $C$ are removed in steps 2a and 2b. Thus, the subroutine returns a \sgec $D$ such that $C\subseteq D$. Finally, by Lemma~\ref{lem:existence of max SGEC}, if $C_i\cap C_j\neq \emptyset$, then there exists a \sgec $E$ such that $C_i\cup C_j\subseteq E$. Thus, $E$ is also returned in the list, and will replace $C_i$ and $C_j$.
We conclude that the returned list contains exactly the maximal \sgec of $\M$.

\subsection{Proof of Lemma~\ref{lem:finite memory EC prob 1}}
\label{apx:finite memory EC prob 1}
Let $f$ be a finite-memory strategy with memory $M$. Consider a history $h\in S^*\times S_1$. Let $m\in M$ be the memory element that $f$ reaches after reading $h$, we define the strategy $f_h$ to be $f$ when starting from $m$. Note that the set $F=\set{f_h:h\in S^*\times S_1}$ is finite, since $M$ is finite.
We claim that there exists $g\in F$ that satisfies the conditions of the lemma. 

Indeed, assume by way of contradiction that for every $g\in F$ we have that $\Pr_{\M^s}(\Inf(g)=C)<1$. Thus, there exists $\epsilon>0$ such that $\Pr_{\M^s}(\Inf(g)=C)< 1-\epsilon$ for every $g\in F$. It follows that there exists $\delta>0$ such that for every history $h$, w.p. at least $\delta$ the strategy $f_h$ from $s$ reaches either a state $t\notin C$ or a state $t'\in C$ such that there exists $s'\in C$ that is not reachable from $t'$ under $f_h$. Since $\delta$ is independent of $h$, and since this is true for every $h$, we get that $\Pr_\M(\Inf(f)=C)=0$, in contradiction to the assumption. 

Let $G=\set{g\in F: \Pr_{\M^s}(\Inf(g)=C)=1}$, then we conclude that $G\neq \emptyset$. Assume by way of contradiction that for every $g\in G$ it holds that there exists a play of $g$ from some state $s\in C$ that leaves $C$ (which happens after a finite number of steps). Thus, there exists some $\delta>0$ such that w.p. at least $\delta$ (independent of $g$), for every $g'\in G$ and every state $s\in C$ a play of $g'$ leaves $C$ (since every $g'\in G$ visits every state of $C$ w.p. 1). This contradicts the fact that $\Pr_{\M^s}(\Inf(g)=C)=1$. 
We conclude that there exists $g\in G$ such that every play of $g$ stays in $C$ forever. 

In addition, since $f$ is parity winning, and the parity condition is independent of the history, then $g$ is parity winning too.

\subsection{Proof of Lemma~\ref{lem:finite memory stays in SGEC}}
\label{apx:finite memory stays in SGEC}
Let $g$ be a strategy obtained as per Lemma~\ref{lem:finite memory EC prob 1}. Thus, $\Pr_{\M^s}(\Inf(g)=C)=1$ for every $s\in C$,  every play of $g$ from $s$ stays in $C$, and $g$ is parity winning. We show that $C$ is a \sgec by showing that $g$ is a witness thereof.
Indeed, w.p. 1 $g$ visits every state of $C$, and in particular $g$ reaches $C^{\me}$ w.p. 1. In addition, $g$ is parity-winning, so every play of $g$ is parity winning, in particular plays that do not reach $C^{\me}$.
\qed

\section{Applications}
\label{app app}

\subsection{Automata, and the Boolean Synthesis Problem}
\label{sec:application defs}
\stam{
For finite sets $I$ and $O$ of input and output signals, respectively, an {\em $I/O$  transducer} is $\T=\zug{I,O,Q,q_0,\delta,\rho}$, where $Q$ is a set of states, $q_0 \in Q$ is an initial state, $\delta: Q\times \tIN \to Q$ is a total (deterministic) transition function, and $\rho:Q\to \tOUT$ is a labeling function on the states. The run of $\T$ on a word $w=i_0 \cdot i_1 \cdots \in \tINo$ is the sequence of states $q_0,q_1,\ldots$ such that $q_{k+1} = \delta(q_k,i_{k})$ for all $k \geq 0$. The {\em output} of $\T$ on $w$ is then $o_1,o_2,\ldots\in \tOUTo$ where $o_k=\rho(q_{k})$ for all $k\ge 1$. Note that the first output assignment is that of $q_1$, and we do not consider $\rho(q_0)$. This reflects the fact that the environment initiates the interaction. The {\em computation of $\T$ on $w$\/} is then 
$\T(w)=i_0\cup o_1,i_1\cup o_2,\ldots \in (2^{I \cup O})^\omega$.
When $Q$ is a finite set, we say that the transducer is  finite.

The synthesis problem gets as input a specification $L \subseteq (2^{I \cup O})^\omega$ and generates a transducer $\T$ that realizes $L$; namely, all the computations of $\T$ are in $L$. The language $L$ is typically given by an LTL formula \cite{Pnu81} or by means of an automaton of infinite words. 
}

An {\em automaton\/} is a tuple 
$\A=\langle\Sigma,Q,q_0,$ $\delta,\alpha\rangle$, where $Q$ is a finite set of states, $q_0 \in Q$ is an initial state, $\delta: Q\times \Sigma \to 2^Q$ is a transition function, and 
$\alpha$ is an acceptance condition. We define some acceptance conditions below.
%We sometimes refer to $\delta$ as a relation $\Delta\subseteq Q\times \Sigma\times Q$, with $\zug{q,\sigma,q'}\in \Delta$ iff $q'\in \delta(q,\sigma)$. 
The automaton $\A$ may run on finite or infinite words. 
A run of $\A$ on a finite word $w=\sigma_1 \cdot \sigma_2 \cdots \sigma_n\in \Sigma^*$ is a sequence of states $r=r_0,r_1,\ldots,r_n$ such that $r_{i+1} \in \delta(r_i,\sigma_{i+1})$ for all $0 \leq i < n$. When $w$ is infinite, so is a run of $\A$ on it. For an infinite run $r$, we denote by $\Inf(r)$ the set of states that $r$ visits infinitely often.
 
We consider two acceptance conditions. When $\A$ runs on finite words, we have that $\alpha \subseteq Q$ is a set of accepting states. Then, a finite run $r_0,r_1,\ldots,r_n$ is accepting if $r_n \in \alpha$. When $\A$ runs on infinite words, then $\alpha:Q\to \set{0,...,d}$ is a {\em parity\/} acceptance condition. For a state $q\in Q$, we refer to $\alpha(q)$ as the {\em rank} of $q$. Then, an infinite run $r$ is accepting if $\maxs{\alpha(q): q\in \Inf(r)}$ is even.

%shaull3 - we don't need nondeterminism
%We also consider two branching modes of automata. When an automaton is nondeterministic, 
%a word $w \in \Sigma^\omega$ is accepted by $\A$ if there exists an accepting run of $\A$ on $w$. Dually, 
The automata we consider are universal. Thus, 
%when the automaton is universal, then 
a word $w\in \Sigma^\omega$ is accepted if all the runs of $\A$ on it are accepting.  
The language of $\A$, denoted $L(\A)$, is the set of words that $\A$ accepts. 
If $|\delta(q,\sigma)|=1$ for every $q\in Q$ and $\sigma\in \Sigma$, we say that $\A$ is {\em deterministic}. Note that in this case, $\A$ has exactly one run on every word.

%We use DFW, UPW, and DPW to denote deterministic automata on finite words, universal parity automata, and deterministic parity automata, respectively.

The classical solution to the Boolean synthesis problem proceeds as follows. Consider a specification \DPW
 $\A=\zug{\tIN\times\tOUT,Q,q_0,\delta,\alpha}$. We obtain from $\A$ a parity game
 $\G=\langle Q \times 2^I,Q,q_0,2^O,2^I, $ $\delta_1,\delta_2,\alpha'\rangle$, where $\delta_1(\zug{q,i},o)=\delta(q, i \cup o)$, and $\delta_2(q,i)=\zug{q,i}$. Thus, Player~2, the environment, controls the inputs and his actions correspond to assignments to the input signals. His states are the states of $\A$, and he moves to states that maintain the assignment he gives to the input signals. Then, Player~1, the system, controls the outputs and his actions correspond to assignments to the output signals. He moves in states that maintain the assignment to the input signals given by Player~2, and his transitions update the state of $\A$. Then, $\alpha'$ in induced by $\alpha$. Formally, for every $q \in Q$ and for every $i \in 2^I$, we have that $\alpha'(q)=\alpha'(\zug{q,i})=\alpha(q)$. It is not hard to see that a winning strategy for Player~1 in $\G$ induces a transducer that realizes $\A$ \cite{PR89a}. 
%shaull3 Yak yak This 
Finding a winning strategy for Player~1 amounts to solving a turn-based parity game, whose complexity is NP$\cap$co-NP. Alternatively, deterministic algorithms for solving parity games run in time polynomial in the number of states, and singly-exponential in the number of parity ranks. When the starting point is an LTL formula $\psi$, the translation to a \DPW involves a doubly-exponential blow up, but the index of the \DPW is only exponential, so the problem is 2EXPTIME-complete \cite{Ros92}.

\subsection{Synthesis with Penalties}
\label{app penalties}

Let $\U_i=\zug{\tIN\times\tOUT,Q^i,q^i_0,\delta^i,\alpha^i}$
Let $S=Q \times S_1 \times \cdots S_m$ and $s_0=\zug{q_0,q_0^1,\ldots,q_0^m}$. We define the parity-MDP $\M=\zug{S \times 2^I,S,s_0,2^O,2^I,\delta_1,\delta_2,\MDPProb,\MDPcost,\alpha'}$ where for every $s=\zug{q,q^1,...,q^m}\in S$, $i\in \tIN$, and $o\in \tOUT$, we have the following. The transition functions are $\delta_1(\zug{s,i},o)=\zug{\delta(q,i \cup o),\delta^1(q^1,i \cup o),\ldots,\delta^m(q^m,i \cup o)}$, and $\delta_2(s,i)=\zug{s,i}$, the cost function is given by $\MDPcost(\zug{s,i})=0$ and $\MDPcost(s)=\sum_{j:q^j\in \alpha^j}\gamma(j)$, for the penally function $\gamma$, and the acceptance condition is $\alpha(s)=\alpha(\zug{s,i})=\alpha(q)$. Finally, we assume that the environment behaves uniformly. That is, in every step it outputs every $i\subseteq I$ with probability $2^{-|I|}$. Thus, $\MDPProb(s,i)=2^{-|I|}$. This assumption can easily be replaced by a different probabilistic model. 

It is easy to see that a winning strategy for Player~1 in $\M$ corresponds to a transducer that realizes $\A$, and that the cost of every computation is the average penalty along the computation. Thus, a solution to the synthesis with penalties problem amounts to solving $\M$. The size of $\M$ is polynomial in the size of the automata $\A,\U_1,\ldots \U_m$, and is exponential in $m$. However, we observe that the role of $\U_1,\ldots, \U_m$ is only for the purpose of costs, and does not affect the parity constraints. Thus, we can solve the problem in NP$\cap$co-NP in the size of the automata, and in time singly-exponential in $m$.
Finally, if $\A$ is obtained by translating an LTL formula $\psi$ into a DPW, then similarly to the case of Boolean synthesis, we can solve the problem in times doubly-exponential in the length of $\psi$, polynomial in $\U_1,\ldots,\U_m$, and singly-exponential in $m$.

\subsection{Proof of Theorem~\ref{thm:sensing to parity-MDP}}
\label{app sen construction}

We identify a subset $i\subseteq I$ with its characteristic function $i:I\to \set{0,1}$. 

Consider the $\DPW$ $\A=\zug{\tIN\times \tOUT,Q,q_0,\delta,\alpha}$. We obtain from $\A$ the $\UPW$ $\A'=\zug{\tIN\times \tIN\times \tOUT,Q,q_0,\delta',\alpha}$ with $\delta'$ defined as follows. Consider a letter $\zug{i,x,o}\in \tIN\times \tIN\times \tOUT$. We think of $i$ and $o$ as truth assignments for the input and output signals, respectively, and we think of $x$ as a set of sensed signals. Consider the set $\sfrac{i}{x}=\set{j\in \tIN: \forall p\in x,\ j(p)=i(p)}$. Intuitively, $\sfrac{i}{x}$ is the set of input assignments that agree with $i$ on all the signals in $x$. For a state $q\in Q$, we define $\delta'(q, \zug{i,x,o})=\set{\delta(q,(j,o)): j\in \sfrac{i}{x}}$. 

Intuitively, when thinking of $\A'$ as a game between the system and the environment, then at each step, the system chooses a set of sensed inputs $x$ and an output $o$. Then, the environment chooses a set of inputs $i$, but in the next step the system can only see the inputs in $i$ that are sensed in $x$, and thus moves universally with every input that agrees with $i$ on the sensed inputs in $x$.

We proceed to determinize $\A'$ to a DPW $\D=\zug{\tIN\times \tIN\times \tOUT,S,\rho,s_0,\beta}$. We then obtain from $\D$ a parity game, as described above, with Player~1 (the system) controlling the set of sensed inputs and the output, and Player~2 (the environment) controlling the concrete inputs. Formally, the game 
$G_\D=\zug{S_1\cup S_2,\start,\Act_1,\Act_2,\delta_1, \delta_2,\beta'}$ is defined as follows. The states are $S_1=(S\times \tIN\times \tIN)\cup \set{\start}$ and $S_2=S\times \tIN$. The actions for Player~1 in every state are $\Act_1=\tIN\times \tOUT$ and are $\Act_2=\tIN$ for Player~2 (we omit the state as the available actions are independent of the state). The transition function is defined as follows. For a state $\zug{s,x,i}\in S_1$ and action $\zug{x',o}\in A_1$ we have $\delta_1(\zug{s,x,i},\zug{x',o})=\zug{\rho(s,\zug{i,x,o}),x'}$ as well as $\delta_1(\start,\zug{x',o})=\zug{s_0,x'}$. 
For a state $\zug{s,x}\in S_2$ and action $i\in A_2$ we have $\delta_2(\zug{s,x},i)=\zug{s,x,i}$. 

Intuitively, the state $\zug{s,x,i}\in S_1$ represents that $\D$ is in state $s$, the system has chosen to sense the signals in $x$, and the environment gave the concrete input $i$. Then, the action $\zug{x',o}$ means that the system responded with output $o$, and chose to sense $x'$ in the next step, taking the game to the state $\zug{s',x'}$, where $s'=\rho(s,\zug{i,x,o})$. Then, in state $\zug{s',x'}$, the environment chooses a new concrete input $i'$.

We define the acceptance condition $\beta'$ as follows. For every $s\in S$ and $i,x\in \tIN$, we have $\beta'(\zug{s,x,i})=\beta'(\zug{s,x})=\beta(s)$, and we arbitrarily set $\beta'(\start)=0$ (since $\start$ is visited only once, this has no effect).

Note that crucially, for every $j,j'\in \sfrac{i}{x}$, the behavior of $G_\D$ from state $\zug{s,x,j}$ is identical to the behavior from $\zug{s,x,j'}$. This follows from the universal transitions in $\A'$. Thus, once Player~1 chooses $x$, the inputs that are not sensed do not play a role. This captures the fact that every winning strategy for the system must only rely on the values $j$ assigns to the sensed inputs $x$.

Finally, the parity-MDP $\M$ is obtained from $G_\D$ by fixing Player~2 with a uniform-stochastic strategy and adding costs according to the number of sensed inputs at each state. Recall that the actions of Player~2 are $\tIN$. Thus, in state $\zug{s,x}\in S_2$, the probability of Player~2 playing $j\in \tIN$ is $2^{-|I|}$. Note that by our observation above, every $j,j'\in \sfrac{i}{x}$, induce the same transitions. Thus, the probability of transition from state $\zug{s,x}$ to $\zug{s,x,j}$ is $2^{-|x|}$.

The cost function assigns cost $|x|$ to states $\zug{s,x}$ and $\zug{s,x,j}$, for every $s\in S$ and $j\in \tIN$.

We now proceed to analyze the correctness of the construction. Consider a (not necessarily finite) transducer $\T=\zug{I,O,T,t_0,\tau,\rho}$ that realizes the specification $\A$. We identify with $\T$ a strategy $f_\T$ for $\M$ as follows. 
In state $\start$ we have $f_\T(\start)=\zug{\sen(t_0),\rho(t_0)}$. Then, the strategy $f_\T$ keeps track of the state of $\T$ as follows. When $\T$ is in state $t$, and the state of the game is $\zug{s,x,i}$, let $t'=\tau(t,i)$. Then, we have that $f_\T(\zug{s,x,i})=\zug{\sen(t'),\rho(t')}$. Observe that $f_\T$ is essentially implemented by the transducer $\T$. In particular, if $\T$ has finite state space, then $f$ has finite memory.

We claim that $\costs(f_\T)=\scost(\T)$. We start by showing that $f_\T$ is sure winning in $\M$ (equivalently, that it is a winning strategy for Player~1 in $G_\D$). Consider an input sequence $\pi\in \tINs$, let $q$ and $t$ be the states that $\A$ and $\T$ reach, respectively, when they interact on $\pi$. let $x=\sen(t)$, then for every $i,j\in \tIN$ such that $j\in \sfrac{i}{x}$ we have that $\tau(t,i)=\tau(t,j)$. Thus, the behavior of $\T$ from $\delta(q,i\cup \rho(t))$ and from $\delta(q,j\cup \rho(t))$ is the same. It follows that $\T$ induces a realizing strategy for the UPW $\A'$ (and hence a winning strategy for $G_\D$), where the additional $\tIN$ component in the alphabet represents the sensing of the current state of $\T$. However, this is exactly the behavior prescribed by $f_\T$, so $f_\T$ is winning in $G_\D$. 

Next, observe that by the above, for every input sequence $\pi\in \tINo$, the (prefix of the) play of $G_\D$ induced by Player~1 playing $f_\T$ and Player~2 playing $\pi$ is $r=\start,\zug{s_1,x_1},\zug{s_1,x_1,\pi_1},...,$ $\zug{s_m,x_m},\zug{s_m,x_m,\pi_m}$, and we have that $\cost_m(f_\T,\pi)=\frac{1}{2m+1}(\sum_{k=1}^m 2\cdot|x_k|)$, while for the run $r=t_1,t_2,...,t_m$ of $\T$ on the first $m$ letters of $\pi$ we have that $\scost(r)=\frac{1}{m}\sum_{k=1}^m \scost(t_k)$. By the definition of $\T$ and $\M$, we have $\scost(t_k)=|x_k|=\MDPcost(\zug{s_k,x_k})=\MDPcost(\zug{s_k,x_k,\pi_k})$. Moreover, the probabilities of $\M$ imply that every $\pi$ such that $|\pi|=m$ is played w.p. $|\tIN|^{-m}$. Thus, by taking $m\to \infty$, we get $\costs(f_\T)=\scost(\T)$. 

Since this is true for every realizing transducer $\T$, it follows that $\costs(\M)\le\scost_{I/O}(L(\A))$.

Conversely, consider a strategy $f$ for $\M$. A-priori, $f$ can behave differently in states $\zug{s,x,i}$ and $\zug{s,x,j}$ for $j\in \sfrac{i}{x}$. However, as we observed above, the construction of $\A'$ (and thus of $\D$) implies that $f$ cannot decrease its cost by doing so, since the behavior of $\A'$ is the same in both states. Thus, we can assume w.l.o.g that $f$ only depends on the values $i$ assigns to the sensed inputs $x$. Now, $f$ induces a (possibly infinite) transducer $\T_f$ in an obvious manner - whenever $f$ outputs $\zug{x,o}$, the transducer outputs $o$. Similar arguments as the converse direction show that $\costs(f)=\scost(\T_f)$, and thus $\costs(\M)\ge\scost_{I/O}(L(\A))$, and we are done.

\end{document}